%% file: MUD_journal.tex
\documentclass[letterpaper, journal,draft, 11pt, onecolumn]{IEEEtran}

\usepackage[final]{graphicx}
\usepackage{graphics}
\usepackage{amsmath,amsthm,amssymb,algorithmic}
\usepackage{algorithm,bm}
\usepackage{subfig}
\usepackage{color}%
\usepackage{setspace}
\usepackage{epstopdf}
\usepackage{soul}

\usepackage{MnSymbol}

\input defs.tex

\graphicspath{{graphs/}}

\usepackage{verbatim}

\title{Compressive Demodulation of Mutually Interfering Signals}

\author{Yuejie Chi, Yao Xie and Robert Calderbank \\
\thanks{Y. Chi is with the Department of Electrical and Computer Engineering, The Ohio State University, Columbus, OH 43210, USA (email: chi@ece.osu.edu).}
\thanks{Y. Xie is with the Department of Electrical and Computer Engineering, Duke University, Durham, NC 27708 (email: yao.xie@duke.edu).}
\thanks{R. Calderbank is with the Department of Computer Science, Duke University, Durham, NC 27708, USA (email: robert.calderbank@duke.edu). }

\thanks{The work of Y. Chi and R. Calderbank was supported by ONR under Grant N00014-08-1-1110, by AFOSR under Grant FA 9550-09-1-0643, and by NSF under Grants NSF CCF -0915299 and NSF CCF-1017431. The work of Y. Xie is supported by DARPA under Grant N66001-11-4002, MSEE under Grant  FA8650-11-1-7150, and ARO under Grant W911NF-09-1-0262.}
\thanks{This paper was presented in part at the 2012 International Statistical Signal Processing Workshop (SSP) and the 2012 Allerton Conference on Communication, Control, and Computing.}
}

\date{\today}

\begin{document}

\maketitle
\begin{abstract}
Multi-User Detection is fundamental not only to cellular wireless communication but also to Radio-Frequency Identification (RFID) technology that supports supply chain management. The challenge of Multi-user Detection (MUD) is that of demodulating mutually interfering signals, and the two biggest impediments are the asynchronous character of random access and the lack of channel state information. Given that at any time instant the number of active users is typically small, the promise of Compressive Sensing (CS) is the demodulation of sparse superpositions of signature waveforms from very few measurements. This paper begins by unifying two front-end architectures proposed for MUD by showing that both lead to the same discrete signal model. Algorithms are presented for coherent and noncoherent detection that are based on iterative matching pursuit. Noncoherent detection is all that is needed in the application to RFID technology where it is only the identity of the active users that is required. The coherent detector is also able to recover the transmitted symbols. It is shown that compressive demodulation requires $\mathcal{O}(K\log N(\tau+1))$ samples to recover $K$ active users whereas standard MUD requires $N(\tau+1)$ samples to process $N$ total users with a maximal delay $\tau$. Performance guarantees are derived for both coherent and noncoherent detection that are identical in the way they scale with number of active users. The power profile of the active users is shown to be less important than the SNR of the weakest user. Gabor frames and Kerdock codes are proposed as signature waveforms and numerical examples demonstrate the superior performance of Kerdock codes - the same probability of error with less than half the samples.
\end{abstract}


\section{Introduction}

Demodulation of mutually interfering signals, or Multi-User Detection (MUD) is central to multiaccess communications \cite{verduMUD1998}. It includes the special case of the ``on-off'' Random Access Channel (RAC) \cite{FletcherRanganGoyal2010} that arises in modeling control channels in wireless networks, where active users transmitting their signature waveforms can be modeled as sending $1$'s to the {Base Station (BS)}, and inactive users can be modeled as sending $0$'s. It also includes the special case of the Radio-Frequency Identification (RFID) system \cite{finkenzeller2010rfid} that arises in supply chain management, where each RFID tag is associated with a unique ID and attached to a physical object. In large scale RFID applications, an {RFID reader} interrogates the environment and all tags within its operational range can be modeled as sending $1$'s, and tags outside its operational range can be modeled as sending $0$'s. It also includes the special case of neighbor discovery in wireless ad-hoc networks \cite{zhang2011wireless,zhang2012neighbor}, where neighbors of a {query node} transmitting their identity information can be modeled as sending $1$'s, and nonneighbors can be modeled as sending $0$'s. In all examples, the received signals are possibly corrupted by noise. 

State-of-the-art random access protocols, such as IEEE 802.11 standards, rely on retransmission with random delays at each active user to avoid collisions. This accumulates to significant delays as the size of the networks becomes large, for example the scale of RFID tags can easily grow to millions in practice. Therefore it is of great interest to allow multiple active users transmit simultaneously and still be able to recover the active users albeit collisions.
The MUD problem becomes the recovery of the active users, and it may be expanded to demodulation of transmitted symbols from each active user in cellular communications. The two biggest impediments are the asynchronous character of random access and the lack of Channel State Information (CSI) at the receiver. The signature waveforms of different users are obtained by modulating a chip waveform using a digital sequence of length $L$. The total number of users $N$ is severely constrained if all signature waveforms are orthogonal, giving the relationship $N\leq L$. In this paper we are interested in both coherent detection when CSI is known and noncoherent detection when CSI is unknown, under the conditions that the signature waveforms are nonorthogonal and the delays of each user are unknown. 

\subsection{Main Contributions}
Our contributions in this paper are three-fold. Given that at any time instant the number of active users $K$ is typically small, the promise of Compressive Sensing (CS) \cite{CandesTao2006,Donoho2006} is the demodulation of sparse superpositions of signature waveforms from very few measurements. A baseline architecture for MUD is correlation of the received signal with a bank of matched filters \cite{verduMUD1998}, each with respect to a shift of a signature waveform. The first drawback is the huge number of required filters, thus the required number of samples, when the number of total users $N$ is large, which is $N_\tau = N (\tau+1)$ where $\tau$ is the maximum delay. A second drawback is that the noise will be colored and amplified by the cross-correlations of selected signature waveforms. An alternative baseline architecture is sampling the received analog signal directly at the chip rate \cite{ApplebaumBajwaDaurteCalderbank2011}. This approach does not amplify the noise but it does require a high-rate Analog-to-Digital Converter (ADC). 

We first demonstrate two front-end architectures for compressive demodulation which can lead to mathematically equivalent discrete signal models. The first architecture is based on subsampling the received signal uniformly at random, which reduces the required rate of ADC in \cite{ApplebaumBajwaDaurteCalderbank2011}. The second architecture is based on a bank of generalized matched filters, which is the extension to asynchronous communication of the architecture for synchronous MUD proposed by Xie et. al. \cite{XieEldarGoldsmith2011} based on analog compressed sensing \cite{MishaliEldar2010}. The novelty is that both architectures are unified under the same discrete signal model, and further reduce the number of acquired samples $M$ to be smaller than the length of the signature waveforms $L$.

Second we present architectures for coherent and noncoherent detection, designed to recover active users and transmitted (QPSK) symbols when the CSI is known, and to recover active users when the CSI is unknown. Both algorithms are based on iterative matching pursuit \cite{Tropp_OMP} and assume a flat-fading channel model so that each active user arrives at the receiver on a single path with an \textit{unknown} delay. We note that the generalization to a small number of arrival paths with a prescribed delay pattern is straightforward. Noncoherent detection is more pertinent to applications like RFID and wireless ad hoc networks, where only identification of active users is of interest. Our main theoretical contribution is relating the probability of error for the proposed MUD algorithms to two geometric metrics associated with the set of subsampled signature waveforms. These metrics, the worst case and average coherence, were introduced by Bajwa et. al. in the context of model selection \cite{BajwaCalderbankJafarpour2010}. We provide explicit performance guarantees in terms of these coherence metrics and the distribution of received signal powers. These fundamental limits quantify robustness of the compressive MUD algorithms to the ``near-far'' problem \cite{verduMUD1998} in multiple access communications. It is shown our proposed compressive MUD algorithms require $\mathcal{O}(K\log N_{\tau})$ samples to recover $K$ active users for both coherent and noncoherent detection, whereas standard MUD requires $N_\tau$ samples. We further show that the minimum signal-to-noise ratio dictated by the weakest active user, rather than the power profile of all active users, plays an important role in the performance of the proposed iterative algorithms; therefore power control is less critical.

Finally, we propose deterministic designs of cyclic-extended signature waveforms  that satisfy both the geometric metrics linked to the decoding algorithms and the block-circulant structure due to cyclic extensions from the asynchronous character. Gabor frames and Kerdock codes are considered due to their optimal coherence properties proved in \cite{BajwaCalderbankJafarpour2010,CalderbankJafarpour2010}, and in this paper we extend this analysis to the uniformly random subsampled Gabor frames and Kerdock codes. Gabor frames are block circulant from its construction as a time-frequency expansion of a seed sequence. The Kerdock code is an extended cyclic code over $\mathbb{Z}_4$ (Section IV, \cite{CalderbankZ4}) and can be arranged to exhibit a block-circulant structure. We demonstrate through numerical simulations that the performance of the proposed compressive MUD algorithms using Gabor frames and Kerdock codes. The superior performance of Kerdock code is emphasized for practical interests, which can obtain the same probability of error with less than half the samples.

\subsection{Relationship to Prior Work}
Here we describe how this paper differs from previous papers that have also formulated MUD as a compressive sensing problem. The focus of most prior work is on synchronous communication, including \cite{FletcherRanganGoyal2010, zhang2011wireless,zhang2012neighbor,XieEldarGoldsmith2011, ZhuGiannakis2011,JinKimRao2010}. In \cite{FletcherRanganGoyal2010}, Fletcher et. al. studied MUD in the context of on-off RACs; in \cite{zhang2011wireless,zhang2012neighbor}, Zhang et. al. studied MUD in the context of neighbor discovery in wireless ad hoc networks; in \cite{XieEldarGoldsmith2011}, Xie et. al. studied MUD with simultaneous symbol detection in cellular communications. The synchronous model provides insight into what might be possible but it ignores the difficulty in estimating the delays of individual users and in achieving synchronization.

A more general asynchronous model is considered by Applebaum et. al. in   \cite{ApplebaumBajwaDaurteCalderbank2011}. These authors assume synchronization at the chip or symbol level, different signature waveforms arrive with different discrete delays in some finite window, and the receiver uses convex optimization to recover the constituents of the sparse superposition. Thus users are associated with a Toeplitz block in the measurement matrix populated by allowable shifts in the signature waveform. In this paper we introduce a cyclic prefix in order to create a measurement matrix with a block cyclic structure which makes it easier to design codebooks using Gabor frames and Kerdock codes.

The algorithms presented in this paper are based on iterative matching pursuit and for uniformly random delays the number of samples they require is of the same order, $\mathcal{O}(K\log N_{\tau})$, as the number required by the convex optimization algorithm presented in \cite{ApplebaumBajwaDaurteCalderbank2011}. This scaling is a significant improvement over the Reduced-Dimension Decision Feedback (RDDF) detector described in \cite{XieEldarGoldsmith2011} which requires order $\mathcal{O}(K^2\log N_{\tau})$ samples. The reason that we are able to break the \textit{square-root bottleneck} is that by introducing more sophisticated coherence metrics we are able to treat average case rather than worst case performance. These methods may be of independent interest. Note also that the complexity of our algorithms are significantly less than that of of convex optimization when the set of active users is highly sparse ($K\ll N_\tau$) \cite{TroppGilbert}. Moreover, it is possible to further reduce the complexity by terminating the algorithm early and obtaining partial recovery of active users. When the channel is known at the receiver we also improve upon the transmission rate reported by Xie et. al. \cite{XieEldarGoldsmith2011} by incorporating complex channel gains in our model and moving from BPSK to QPSK signaling.


Our focus on deterministic signature waveforms is different from most previous work \cite{FletcherRanganGoyal2010,zhang2011wireless,zhang2012neighbor} which considers random waveforms. The fact that random waveforms can be shown to satisfy the Restricted Isometry Property \cite{CandesTao2006} makes analysis possible but they are not very practical. The same criticism can be leveled at the RDDF detector described in \cite{XieEldarGoldsmith2011} where randomness enters the choice of the coefficients determining the filter bank. Randomness also enters into \cite{zhang2012neighbor} through the pattern of puncturing of Reed-Muller codewords which serve as deterministic signature waveforms.

\subsection{Organization of this paper and Notations}
The rest of the paper is organized as follows. Section \ref{sec:model} describes the system model, and Section \ref{sec:frontend} presents two architectures for the compressive MUD front-end. Section \ref{sec:algorithms} proposes the coherent and noncoherent detectors, along with their performance guarantees. Section \ref{sec:proofs} proves the main theorems. Section \ref{sec:waveforms} presents the design of signature waveforms based on Gabor frames and Kerdock codes. Section \ref{sec:numerical} shows the numerical simulations and Section \ref{sec:conclusion} concludes the paper.

Throughout the paper, we use capital bold letters $\A$ to denote matrices, small bold letters $\ab$ to denote vectors, $\|\A\|_{p}$ and $\|\ab\|_{p}$ to denote the $p$-norm of $\A$ and $\ab$, where $p=2$ or $\infty$. $\I_N$ denotes the identity matrix of dimension $N$, $\dag$ denotes pseudo-inverse, $\A^H$ denotes the Hermitian of $\A$, and $c^*$ defines the conjugate of a complex number $c$.

\input{model.tex}

\input{frontend.tex}

\input{algorithms.tex}

\input{proofs.tex}

\input{waveforms.tex}

\input{numerical.tex} 
\section{Conclusions}\label{sec:conclusion} 
This paper describes two MUD front-end architectures that lead to mathematically equivalent discrete signal models. Both coherent and noncoherent detectors based on iterative matching pursuit are presented to recover active users, and their transmitted symbols are also detected in the coherent case. It is shown that compressive demodulation requires $\mathcal{O}(K\log N_\tau)$ samples to recover $K$ active users. Gabor frames and Kerdock codes are proposed as signature waveforms and numerical examples are provided where the superior performance of Kerdock code is emphasized. The resilience of iterative matching pursuit to variability in relative strength of the entries of the signal might be an advantage in multi-user detection in wireless communications because it makes power control less critical. We make the final remark that the noncoherent detectors can be extended to detect transmitted symbols by assigning two different signature waveforms to the BPSK signaling.

\section*{Appendix}

\subsection{Sidak's lemma}
\begin{lemma}[Sidak's lemma] \cite{Sidak}\label{sidaklemma}
Let $[X_1,\cdots, X_n]$ be a vector of random multivariate normal variables with zero means, arbitrary variances $\sigma_1^2$, $\cdots$, $\sigma_n^2$ and
and an arbitrary correlation matrix. Then, for any positive numbers $c_1, \cdots, c_n$, we have
$$ \Pr( |X_1| \leq c_1, \cdots,  |X_n| \leq c_n) \geq \prod_{i=1}^n  \Pr( |X_i| \leq c_i). $$
\end{lemma}

\subsection{Proof of Lemma~\ref{noise}} \label{proof_noise}
Since $\wb\sim\mathcal{CN}(\mathbf{0},\sigma^2\I_M)$, $\X^H\Pb\wb\sim\mathcal{CN}(\mathbf{0},\sigma^2\X^H\Pb\X)$ but it is a colored Gaussian noise. We want to bound $\Pr(\| \X^H\Pb\wb \|_\infty\geq \tau)$ for some $\tau>0$. Note that each $\xb_n^H\Pb\wb\sim\mathcal{CN}(0,\sigma_n^2)$, where $\sigma_n^2=\sigma^2 \xb_n^H\Pb \xb_n\leq \sigma^2$ (recall that $\|\xb_n\|_2 = 1$). Then
$$ \Pr(| \xb_n^H\Pb\wb | \leq \tau) =1-\frac{1}{\pi}e^{-\tau^2/\sigma_n^2}\geq 1-\frac{1}{\pi}e^{-\tau^2/\sigma^2}. $$
Following Lemma \ref{sidaklemma}, for $\tau>0$ we have
\begin{align*}
\Pr(\| \X^H\Pb\wb \|_\infty \leq \tau) & \geq \prod_{n=1}^{N_\tau} \Pr(| \xb_n^H\Pb\wb | \leq \tau) \\
& \geq (1-\frac{1}{\pi}e^{-\tau^2/\sigma^2})^{N_\tau}  \geq 1-\frac{N_\tau}{\pi}e^{-\tau^2/\sigma^2},
\end{align*}
provided the right hand side is greater than zero.

\subsection{Proof of Lemma~\ref{lemma_active_user}}\label{proof_active_user}

We begin by deriving a lower-bound for $\max_{n\in\mathcal{I}}|\xb_n^H \yb|$ when $\mathcal{G}_1\cap\cH_0$ occurs. Assume that $n_0$ is the index achieving the largest absolute gain: $|r_{n_0}| = |r|_{(1)}$. Then under the event $\mathcal{G}_1\cap\cH_0$:
\begin{align}
\max_{n \in \mathcal{I}} |\xb_n^H \yb|  \geq |\xb_{n_0}^H \yb| 
&= \left|b_{n_0} r_{n_0} + \sum_{m\neq n_0} b_m r_m \xb_{n_0}^H \xb_m + \xb_{n_0}^H \wb \right|  \nonumber \\
& \geq |r|_{(1)}  - \left| \sum_{m\neq n_0} b_m r_m \xb_{n_0}^H \xb_m\right| -  |\xb_{n_0}^H \wb|  \nonumber  \\
&= |r|_{(1)} - \|(\X_{\Pi}^H \X_\Pi-\I)\rb_{\mathcal{I}}\| - |\xb_{n_0}^H \wb|  \nonumber \\
& >  |r|_{(1)}| - \epsilon\|\rb_{\mathcal{I}}\|_2 -  \tau. \label{omp_1}
\end{align}


%
On the other hand, we can similarly expand and upper-bound $\max_{n\notin\mathcal{I}}|\xb_n^H\yb|$, under the event $\mathcal{G}\bigcup\mathcal{G}_{\mathcal{I}}$, as
\begin{align}
\max_{n\notin \mathcal{I}} |\xb_n^H \yb|
&=\max_{n\notin \mathcal{I}} \left|\sum_{m\in\mathcal{I}}b_m r_m \xb_n^H \xb_m + \xb_n^H \wb \right| \nonumber \\
&\leq \max_{n\notin \mathcal{I}} \left|\sum_{m\in\mathcal{I}}b_m r_m \xb_n^H \xb_m\right| +  \max_{n\notin \mathcal{I}}|\xb_n^H \wb| \nonumber \\
&=  \|\X_{\Pi^c}^H\X_\Pi \zb\|_\infty + \max_{n\notin \mathcal{I}}|\xb_n^H \wb| \nonumber\\
& <  \epsilon\|\rb_{\mathcal{I}}\|_2 + \tau.\label{ineq_2}
\end{align}

Combining \eqref{omp_1} and \eqref{ineq_2}, we have that under the event $\cG_1\cap\cH_0$,
\begin{equation}
\max_{n \in \mathcal{I}} |\xb_n^H \yb| > |r|_{(1)} - 2\epsilon\|\rb_{\mathcal{I}}\|_2 - 2\tau + \max_{n\notin\mathcal{I}}|\xb_n^H \yb|.
\end{equation}
So when $\mathcal{G}$ occurs, under the condition \eqref{cond_2}, we obtain \eqref{rank_OMP}, as required.  

Furthermore, to detect correctly, for $\Re[b_n] = 1/\sqrt{2}$, $\Re[r_n^H \xb_n^H \yb]$ has to be positive, and for $\Re[b_n] = -1/\sqrt{2}$, $\Re[r_n^* \xb_n^H \yb]$ has to be negative. Similarly we can detect $\Im[b_n]$. First assume $\Re[b_n] = 1/\sqrt{2}$,  then 
\begin{align*}
&\quad \Re[r_{n_1}^* \xb_{n_1}^H\yb] \\
& =  |r_{n_1}|^2  + \sum_{m\neq n_1} \Re[b_m] \Re\left[r_{n_1}^* r_m \xb_{n_1}^H \xb_m\right] + \Re\left[r_{n_1}^*\xb_{n_1}^H \wb\right]   
\end{align*}
must be positive. Suppose this does not hold, and $\Re[r_{n_1}^* \xb_{n_1}^H\yb]<0$.
Recall that $n_0$ is the index of the largest gain: $|r|_{n_0}= |r|_{(1)}$. From \eqref{n_1}, we have
\begin{equation}
|r_{n_1}^*\xb_{n_1}^H\yb| \geq |r_{n_1}^*\xb_{n_0}^H\yb|. \label{key_OMP}
\end{equation}
Since
\begin{align}
 |r_{n_1}^*\xb_{n_1}^H \yb| & =\left| |r_{n_1}|^2 +   \sum_{m\neq n_1} b_m r_{n_1}^* r_m \xb_{n_1}^H \xb_m + r_{n_1}^*\xb_{n_1}^H \wb \right|  \nonumber \\
& \leq \left| \sum_{m\neq n_1} b_m r_{n_1}^* r_m \xb_{n_1}^H \xb_m + r_{n_1}^*\xb_{n_1}^H \wb \right| \label{eqn65}\\
&\leq |r_{n_1}|(\epsilon\|\rb_\cI\|_2 + \tau), \nonumber
\end{align} 
where \eqref{eqn65} follows from $\Re[r_{n_1}^* \xb_{n_1}^H\yb]<0$. Similarly to earlier derivations, we have
\begin{align}
 |r_{n_1}^*\xb_{n_0}^H \yb|&> |r_{n_1}| (|r|_{(1)} - \epsilon \|\rb_\mathcal{I}\|_2  - \tau )
\label{eqn86},
\end{align}
we have that once \eqref{cond_2} holds, $ |r_{n_1}^*\xb_{n_0}^H \yb|>|r_{n_1}^H\xb_{n_1}^H \yb|$, which contradicts \eqref{key_OMP}, then $\sign(\Re[r_{n_1}^*\ab_{n_1}^H \yb]) = 1$. A similar argument can be made for $\Re[b_{n_1}] = -1/\sqrt{2}$ and the cases associated with $\Im[b_{n_1}]$, which completes the proof.

\subsection{Proof of Proposition~\ref{subsample_coherence}} \label{proof_coherence}

Denote the index set of subsampled rows of the Gabor frame as $\Lambda$. Let $\phi_m(i)$ be the $i$th entry of $\mphi_m$, the coherence between two distinct columns of $\X$ is given as $m\neq m'$, 
\begin{align*}
\langle \xb_m, \xb_{m'} \rangle &=\sum_{i\in\Lambda}  \phi^*_m(i)\phi_{m'}(i) , 
\end{align*}
with the expectation $\mathbb{E}\langle \xb_m, \xb_{m'} \rangle  = \langle \mphi_m, \mphi_{m'}  \rangle$, whose absolute value is upper bounded by $ \mu(\mPhi)$, the worst case coherence of $\mPhi$. Applying the triangle inequality and the Hoeffding's inequality \cite{Hoeffding} we have for $\gamma>0$,
\begin{equation*}
\Pr\left\{ \left| \langle \xb_m, \xb_{m'} \rangle \right| - \mu(\mPhi) \geq \gamma \right\}  \leq 4\exp \left(-\frac{\gamma^2 M}{4} \right),
\end{equation*}
Now we consider all pairs of different inner products and apply the union bound,
\begin{align*}
\Pr\left\{ \mu(\X)- \mu(\mPhi) \geq \gamma \right\}  & \leq 2P^2(P^2-1)\exp \left(-\frac{\gamma^2 M}{4} \right) \\
& < 2P^4\exp \left(-\frac{\gamma^2 M}{4} \right).
\end{align*}
Let $\gamma=\sqrt{\frac{20\log P}{M}}$, then with probability at least $1-2P^{-1}$, we have 
\begin{equation} \label{random_coherence}
\mu(\X)\leq  \mu(\mPhi)+\sqrt{\frac{20\log P}{M}}.
\end{equation}

If the Gabor frame satisfies the coherence property such that $\mu(\mPhi)\leq \gamma_1/\log P$ for some constant $\gamma_1$, then by choosing $M\geq \gamma\log^3 P$, we have $\mu(\X)<\gamma_2/\log P$ for some constant $\gamma_2$ with probability at least $1-2{P}^{-1}$.

We next consider the average coherence of $\X$. Let $m=Pq+r$, $m'=Pq'+r'$, we have
\begin{equation*}
\sum_{m' \neq m} \langle \xb_m, \xb_{m'} \rangle = \sum_{i\in\Lambda} \sum_{m' \neq m}\phi^*_m(i)\phi_{m'}(i).
\end{equation*}
Since each column in a Gabor frame can be written as,
$$ \mphi_{m} = [g_{(1-r)_P}e^{j2\pi\frac{q}{P}\cdot 0}, \ldots, \;g_{(P-r)_P}e^{j2\pi\frac{q}{P}\cdot (P-1)}]\transpose, $$
where $(1-r)_P=\mbox{mod}(1-r, P)$. If $r\neq r'$, we have 
$$\sum_{q'=0}^{P-1}\sum_{r\neq r'}\phi^H_m(i)\phi_{m'}(i)= Pg^{*}_{(1-r)_P}\sum_{r\neq r'} g_{(1-r')_P} \cdot \delta_{\{i=1\}}$$
If $r=r'$, $q\neq q'$, we have 
\begin{equation}
\quad \sum_{q'\neq q}\phi^H_m(i)\phi_{m'}(i) 
=\frac{1}{M}\left[(P-1)\cdot \delta_{\{i=1\}} -  \delta_{\{i\neq 1\}}\right], \nonumber
\end{equation}
where we use the fact $|g_i|^2=1/M$. To sum up, we have
 \begin{equation}\label{term}
\begin{split}
&\quad \sum_{m'\neq m}\phi^H_m(i)\phi_{m'}(i)\\
&= \left\{ \begin{array}{ll}
Pg^{*}_{(1-r)_P}\sum_{r\neq r'} g_{(1-r')_P} +(P-1)/M &\; i = 1 \\
  -1/M & \; i\neq 1
       \end{array} \right. 
\end{split} \nonumber
\end{equation} 
 Then $\left|\sum_{m'\neq m} \langle \xb_m, \xb_{m'} \rangle\right|\leq \frac{P}{\sqrt{M}}\|\g\|_1+1= \frac{P^2}{M}+1$, and the average coherence $\nu(\X)$ can be bounded deterministically as
 $$\nu(\X)=\frac{P^2+M}{P^2-1}\cdot \frac{1}{M}\leq\frac{2}{M}. $$

\subsection{Proof of Proposition~\ref{prop_kerdock}}\label{proof_kerdock}
The analysis of the worst-case coherence is exactly as in the proof of Proposition~\ref{subsample_coherence} hence is not repeated. Regarding the average coherence, the columns in the Kerdock set $\tilde{\mPsi}$ form an abelian group $\mathcal{G}$ under point-wise multiplication. By the fundamental group property, if every row contains some entry not equal to $1$, then the column group $\cG$ satisfies
$\sum_{g\in\cG} g = 0$. Let $x_m(i)$ and $\psi_m(i)$ be the $i$th entry respectively of $\xb_m$ and $\mpsi_m$. When $i\neq 0$, the subsampled Kerdock set $\X$ then satisfies
\begin{align*}
& \quad \left|\sum_{m'\neq m}  \langle x_m(i), x_{m'}(i) \rangle\right| \nonumber \\
&=  \left|\sum_{m'\neq m}  \langle \psi_{m}(i), \psi_{m'}(i)  \rangle -\sum_{m\in\mPsi_c}\langle \psi_{m}(i), \psi_{m'}(i)  \rangle \right|  \\
&\leq   \left|\sum_{m'\neq m}  \langle \psi_{m}(i), \psi_{m'}(i)  \rangle \right| + \sum_{m\in\mPsi_c} \left|\langle \psi_{m}(i), \psi_{m'}(i)  \rangle \right| \\
& \leq  1-\frac{1}{M}+\frac{P}{M}, 
\end{align*}
 and
$$ \sum_{m' \neq m}  \langle x_m(0), x_{m'}(0) \rangle =  \frac{P^2-P-1}{M}. $$
Then the average coherence is bounded as
$$\nu(\X) =\frac{1}{P^2-1} \left|\sum_{m' \neq m}  \langle  \xb_{m'}, \xb_{m} \rangle \right| \leq \frac{P^2+M-2}{(P^2-1)M}\leq \frac{2}{M}. $$

\bibliography{refs}									

\end{document}

%% file: defs.tex
\newcommand{\vect}{\boldsymbol} 
\newcommand{\mat}{\boldsymbol}
\newcommand{\transpose}{^{\top}}
\newcommand{\sign}{\mbox{sgn}}

\newcommand{\A}{\mat{A}}
\newcommand{\X}{\mat{X}}

\newcommand{\cG}{\mathcal{G}}
\newcommand{\cH}{\mathcal{H}}
\newcommand{\cI}{\mathcal{I}}
\newcommand{\rb}{\vect{r}}

\newcommand{\yb}{\vect{y}}
\newcommand{\s}{\vect{s}}
\newcommand{\n}{\vect{n}}

\newcommand{\ab}{\vect{a}}
\newcommand{\bb}{\vect{b}}

\newcommand{\eb}{\vect{e}}

\newcommand{\Pb}{\mat{P}}

\newcommand{\B}{\mat{B}}

\newcommand{\mH}{\mat{H}}
\newcommand{\F}{\mat{F}}
\newcommand{\R}{\mat{R}}

\newcommand{\T}{\mat{T}}
\newcommand{\fb}{\vect{f}}
\newcommand{\I}{\mat{I}}
\newcommand{\la}{\langle}
\newcommand{\ra}{\rangle}

\newcommand{\mPhi}{\mat{\Phi}}
\newcommand{\mPsi}{\mat{\Psi}}
\newcommand{\mpsi}{\mat{\psi}}
\newcommand{\mphi}{\mat{\phi}}

\newcommand{\bomega}{\mat{\omega}}

\newcommand{\vb}{\vect{v}}

\newcommand{\hb}{\vect{h}}
\newcommand{\g}{\vect{g}}
\newcommand{\wb}{\vect{w}}
\newcommand{\zb}{\vect{z}}

\newcommand{\shiftR}{\vect{S}}
\newcommand{\W}{\vect{W}}

\newcommand{\xb}{\vect{x}}

\newcommand{\ben}{\begin{equation}}
\newcommand{\bs}{\begin{split}}
\newcommand{\es}{\end{split}}

\newcommand{\een}{\end{equation}}
\newcommand{\SNR}{\textsf{SNR}}
\newcommand{\LAR}{\textsf{LAR}}

\newtheorem{mydef}{Definition}
\newtheorem{thm}{Theorem}
\newtheorem{prop}{Proposition}
\newtheorem{corollary}[thm]{Corollary}
\newtheorem{lemma}{Lemma}

%% file: model.tex
\section{System Model} \label{sec:model}

Consider a multi-user system of $N$ total user where the $n$th users, $n = 1, \cdots, N$, communicate using spread spectrum waveform of the form
\ben
x_n(t) = \sqrt{P_n} \sum_{\ell =0}^{L-1} a_{n,\ell} p(t - \ell T_c), \quad t\in [0, T),
\een
where $p(t)$ is a unit-energy pulse $\int |p(t)|^2 dt =1$, $\int p^*(t-\ell T_c)p(t-kT_c) dt = 0$, $(\cdot)^*$ denoting the conjugate operation, for $\ell \neq k$.
The chip duration $T_c$ determines the system bandwidth, $T$ is the symbol duration, $P_n$ denotes the transmit power of the $n$th user, and the spreading codeword
\ben
\ab_n = [a_{n, 0} \quad \cdots \quad a_{n, L-1} ]\transpose, \quad n = 1, \cdots, N,
\een
is the $L$-length (real- or complex-valued) codeword of unit energy $\|\ab_n\|_2 = 1$ assigned to the $n$th user. Typically $L< N$. The notation $\transpose$ denotes transpose of a matrix or vector. 

To simplify the model, we consider a one-shot model, where the user sends one symbol at a time rather than sending a sequence of symbols. The signal at the receiver is given by
\ben
y(t) = \sum_{n=1}^N g_n \sqrt{P_n} \delta_{\{n\in\mathcal{I}\}} b_n x_n(t-\tau_n') + w(t),
\een
where $g_n \in \mathbb{C}$ and $\tau_n' \in \mathbb{R}_+$ are the channel fading coefficient and the continuous delay associated with the $n$th user, respectively. Define the power profile of all users as $\rb=[r_1,\cdots, r_N]\transpose$, where
\begin{equation}\label{powerprofile}
r_n \triangleq g_n \sqrt{P_n}.
\end{equation}
The power profile is determined by the power control at the transmitter and the channel coefficients during transmission, which could take complex values.

We assume Quadrature Phased Shift Keying (QPSK) modulation, where $b_n \in \{(-1-j)/\sqrt{2}, (-1+j)/\sqrt{2}, (1-j)/\sqrt{2}, (1+j)/\sqrt{2}\}$ is the transmitted symbol of the $n$th user, and $w(t)$ is a complex additive white Gaussian noise (AWGN) introduced by the receiver circuitry with zero mean and variance $\sigma_0^2$. Denote by $\mathcal{I}$ the set of active users. We assume the support of active users $\mathcal{I}$ is a uniform random $K$-subset of $\lsem N\rsem\triangleq \{1, \dotsc, N\}$.  The Dirac function $\delta_{x} = 1$ if $x$ is true and $\delta_x=0$ otherwise. 

Define the individual discrete delays $\tau_n \triangleq \lfloor \tau_n'/T_c \rfloor \in \mathbb{Z}_+$, and the maximum discrete delay $\tau \triangleq \max_n \tau_n \in \mathbb{Z}_+$.  While the values of $\tau_n$ are unknown, $\tau$ is assumed to be known by the transmitters and receivers. 

Each vectors $\ab_n$ is the cyclic prefix of a vector $\tilde{\ab}_n$ of length $P= L-(\tau+1)$. As shown in Fig.~\ref{CyclicP}, $\ab_n$ is obtained by appending the first $\tau + 1$ symbols of $\tilde{\ab}_n$ to the end of $\tilde{\ab}_n$, we have $\tilde{a}_{n,\ell} = {a}_{n,P-\tau-\ell+l}$ for $\ell=1,\dotsc,\tau+1$. As a result, any length $P$ sub-sequence of the vectors $\ab_n$ will be a cyclic shift of $\tilde{\ab}_n$. 

\begin{figure}[htp]
\centering
\includegraphics[width=0.25\textwidth]{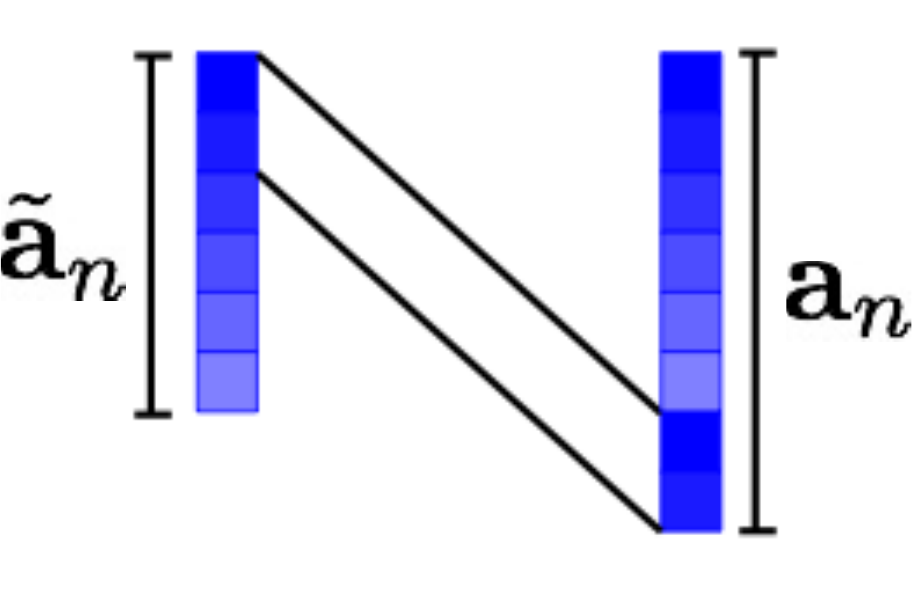} 
\caption{Illustration of the cyclic prefix in the construction of spreading codewords.} \label{CyclicP}
\end{figure}

%% file: frontend.tex
\section{Compressive MUD Front-End}\label{sec:frontend}
In this Section we describe two front-end architectures for compressive MUD. The first is the chip-rate subsampling architecture considered in \cite{XCAC}; and the second is the asynchronous case of a bank of generalized matched filters architecture considered in \cite{XieEldarGoldsmith2011}. We begin by showing that mathematically, the two front-end architectures are equivalent.

\subsection{Chip-rate subsampling architecture}

The chip-rate subsampling architecture directly samples the continuous received signal at the chip rate using a high-rate ADC as shown in Fig.~\ref{Fig:sampling} (a). The receiver only starts sampling when the waveforms of all active users have arrived. Starting at sample $(\tau + 1)$, it collects $M$ uniformly random samples over a window of length $L$. These samples, or linear combinations thereof constitute the measurements made by the receiver. We assume the codewords are of a reasonable length relative to the delays such that $L > M$.  As a result, the output data vector can be written as
\ben
\yb = \bar{\mH} \I_\Omega\A\R\bb + \wb, \label{sig_model}
\een
where $\yb\in\mathbb{C}^{M\times 1}$,  $\A\in\mathbb{C}^{P\times N_{\tau}}$, and the noise $\wb\in\mathbb{C}^{M\times 1}$ is complex Gaussian distributed with zero mean and variance $\sigma_0^2 \bar{\mH}\bar{\mH}^H$.  The subsampling matrix is defined as $\I_\Omega\in\mathbb{R}^{M\times P}$, where $\Omega$ denotes indices of samples, and $\bar{\mH}\in\mathbb{C}^{M\times M}$ is a matrix that linearly combines the samples. The columns of matrix $\A$ have a block structure with each block consisting of circulant shifts of a codeword. Define a circulant matrix $\A_n$ as
\ben
\A_n = \left[
\begin{array}{cccc}
\mathcal{T}_0\tilde{\ab}_n & \mathcal{T}_1\tilde{\ab}_n & \cdots& \mathcal{T}_\tau\tilde{\ab}_n 
\end{array}
\right] \in \mathbb{C}^{P\times (\tau + 1)}, 
\een
%
where the notation $\mathcal{T}_k$ denotes the circulant shift matrix by $k$, and
\ben
\A = [\A_1 \quad\cdots\quad \A_N] \in  \mathbb{C}^{P\times N_\tau}. 
\een
The vector $\bb \in \mathbb{C}^{N_\tau}$ contains the transmitted symbols; it is a concatenation of $N$ vectors $\bb'_n$ of length $\tau+1$, each with at most one non-zero entry at the location of $\tau_n$: 
$$b'_{n, m} = b_n \delta_{\{m = \tau_n \}}, \quad m =0,\cdots,\tau.$$ 
The entries $R_{mm}$ of the diagonal matrix $\R \in \mathbb{C}^{N_\tau}$ are a function of the channel gain, the transmitted power, and the transmitted symbols: 
\begin{align}
R_{mm} &=  r_n \delta_{\{m = (n-1)(\tau+1) + \tau_n\}}, \\
& \quad n = 1, \dotsc, N, \quad m = 0, \dotsc, N_\tau-1. \nonumber
\end{align} 
%


\begin{figure}[h]
\centering
\begin{tabular}{c}
\includegraphics[width=0.35\textwidth]{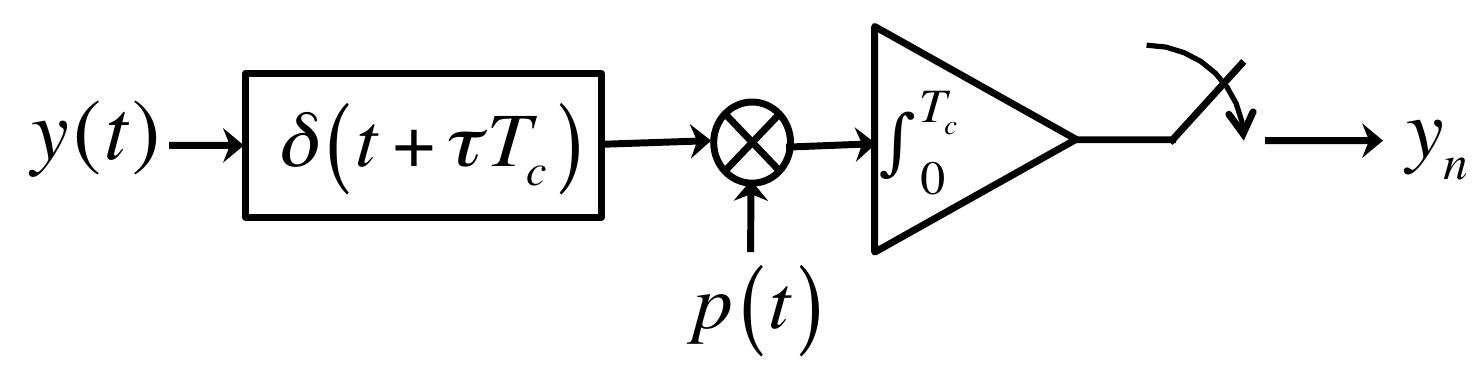} \\
(a) chip-rate subsampling \\
 \includegraphics[width=0.5\textwidth]{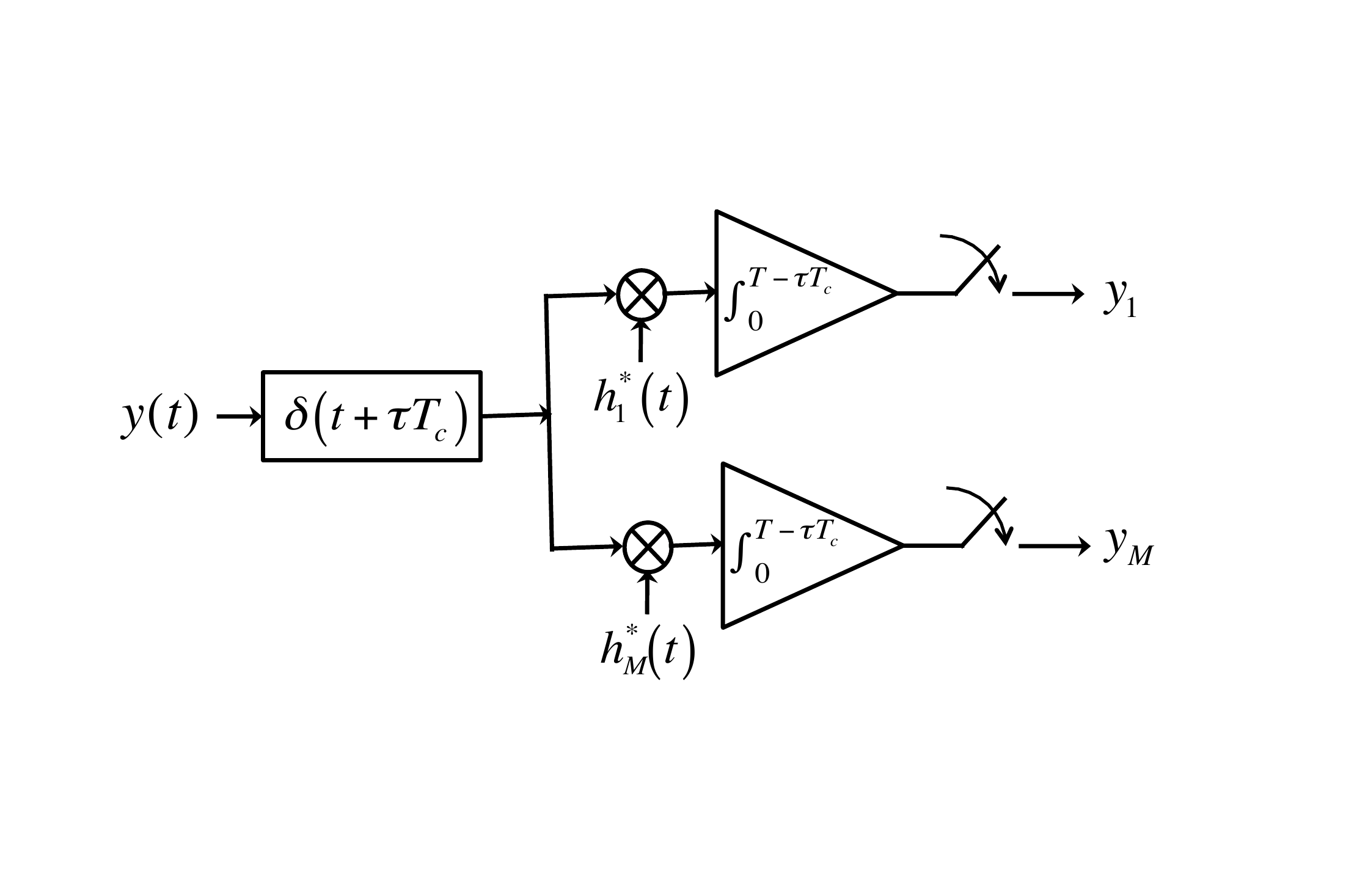} \vspace{-0.4in}\\ 
(b) a bank of generalized matched filters
\end{tabular}
\caption{Illustration of two architectures: (a) the chip-rate subsampling architecture, and (b) a bank of generalized matched filters architecture, where the first block is a linear filter with impulse response $\delta(t + \tau T_c)$. } \label{Fig:sampling}
\end{figure}

\subsection{A bank of generalized matched filters}

A generalized matched filter for compressive MUD \cite{XieEldarGoldsmith2011} correlates $y(t)$ with a set of signals $\{h_m(t)\}_{m = 1}^M$, as shown in Fig. \ref{Fig:sampling} (b). The measurement is taken by multiplying a delayed version of $y(t+\tau T_c)$ with $h_m(t)$, and integrating over a window of length $T-\tau T_c$, where $T$ is the symbol period. 
The output of the $m$th measurement is given by 
\begin{align}
y_m &= \int_0^{T-\tau T_c} h_m^*(t) y(t+\tau T_c) dt \nonumber \\
& \triangleq \la h_m(t), y(t+\tau T_c) \ra,\quad m = 1,\cdots, M. \nonumber
\end{align}
Writing this in a vector notation, we have
\ben\nonumber
\yb = \B \R\bb + \wb, \label{model_B}\een
where 
\ben
\B = [\B_1,  \cdots, \B_N]   \in \mathbb{C}^{M\times N_\tau}, \label{B_block}
\een
with $\B_n\in\mathbb{C}^{M\times (\tau+1)}$, for $n =1, \cdots, N$. The $(m, \ell+1)$th entry of $\B_n$ is given by
\begin{align}
[\B_n]_{m, \ell+1} &= \int_0^{T-\tau T_c} h_m^*(t) x_n(t+\tau T_c-\ell T_c),\nonumber\\
& \quad \ell = 0, \cdots, \tau, \quad m = 1, \cdots, M. \nonumber
\end{align}
The noise vector $\wb$ is a $M$-dimensional complex Gaussian vector with zero mean and covariance matrix 
\begin{equation}
[\vect{\Sigma}]_{mk} = \sigma_0^2 \int_0^{T-\tau T_c} h_m^*(t) h_k(t) dt.
\label{cov_filter}
\end{equation}

We now parameterize for the generalized matched filters $\{h_m(t)\}_{m=1}^M$. In \cite{XieEldarGoldsmith2011}, the matched filters are constructed as linear combinations of the bi-orthogonal signals of the user signature waveforms. Here we consider a more general construction that can lead to a discrete model equivalent to that of the chip-rate subsampling \eqref{sig_model}. Assume the measurement signals are constructed using the chip waveform and chip sequences as
\ben
h_m(t) = \sum_{\ell=0}^{L} h_{m, \ell}p(t - \ell T_c), \label{h_para}
\een
where 
\ben
\vect{h}_m = [h_{m, 0}\quad \cdots \quad h_{m, L-1}]\transpose, \quad m = 1, \cdots, M,
\een
is the $L$-length (real- or complex-valued) codeword 
for the $m$th measurement signal. By this parameterization, for $\ell = 0, \cdots,\tau$,
\begin{align}
&\quad  [\B_n]_{m, \ell+1} \nonumber \\
&=\int_0^{T-\tau T_c} h_m^*(t) x_n(t+(\tau-\ell)T_c)dt \nonumber\\
&= \sum_{u=0}^{L}\sum_{v=0}^{L} 
h^*_{m, u} a_{n, v} \int_0^{T-\tau T_c} p^*(t - uT_c) p(t + (\tau -l -v) T_c) dt \nonumber
\\ 
&= \sum_{u=0}^{L}h^*_{m, u}[ \mathcal{T}_l \tilde{\ab}_{n}]_u  = \vect{h}^H_m(\mathcal{T}_l\tilde{\ab}_n),\label{B_element}
\end{align}
where we have used $\int_0^{T-\tau T_c} p^*(t - uT_c) p(t + (\tau -l -v) T_c) dt = \delta_{\{u=\tau -l -v\}}$. Hence, from (\ref{B_block}) and \eqref{B_element}, we obtain that
\begin{align}
\B & = \begin{bmatrix} \B_1 & \cdots & \B_N  \end{bmatrix}  = \vect{H} \A \label{B_final},
\end{align}
where 
\begin{equation*}
 \B_n = \begin{bmatrix} 
 \vect{h}^H_1(\mathcal{T}_0\tilde{\ab}_n)  & \cdots & \vect{h}^H_1(\mathcal{T}_\tau\tilde{\ab}_n) \\
\vdots &    &\vdots \\ 
 \vect{h}^H_M(\mathcal{T}_0\tilde{\ab}_1)  & \cdots & \vect{h}^H_M(\mathcal{T}_\tau\tilde{\ab}_1) 
  \end{bmatrix}.
\end{equation*}

The noise in the $m$th measurement is given by $\la h_m(t), w(t) \ra$, which is a complex Gaussian random variable with zero mean and covariance matrix (\ref{cov_filter}) given by  $[\vect{\Sigma}]_{mk} = \sigma_0^2 \vect{h}_m^H\vect{h}_k.$
Define a matrix 
\[
\mH = [\hb_1,\cdots, \hb_M]^H \in \mathbb{C}^{M\times P}. 
\]
Substituting (\ref{B_final}) into (\ref{model_B}), we obtain that when the filters $\{h_m(t)\}$ are parameterized by (\ref{h_para}),  the measurement vector can be written as
\begin{equation}
\yb = \mH\A\R\bb + \wb,  \label{mixing}
\end{equation}
where $\wb \sim \mathcal{CN}(\vect{0}, \sigma_0^2 \mH\mH^H)$ is a complex Gaussian random vector with zero mean and covariance matrix $\sigma_0^2 \mH\mH^H$. Given the output (\ref{mixing}) of the bank of generalized matched filters, there are two special cases for $\mH$:
\begin{itemize}
\item $\mH = \bar{\mH}\I_\Omega$, then $\mH\mH^H = \bar{\mH}\bar{\mH}^H$, which means the output \eqref{mixing} of the second architecture is mathematically equivalent to the chip-rate subsampling architecture \eqref{sig_model}. 
\item In the first architecture, if we choose $\bar{\mH}$ to be an orthogonal matrix, then $\bar{\mH}\bar{\mH}^H = \I_M$, the output signal power of each measurement is $M/N$ and the noise power is $\sigma_0^2$. The signal-to-noise ratio per measurement is $M/(N\sigma_0^2)$.
\item In the second architecture, if we choose $\mH$ to be a tight frame, then $\mH\mH^H = (N/M)\I_M$. For each measurement, the output signal power is $1$, and the noise power is $(N/M) \sigma_0^2$. The signal-to-noise ratio per measurement is $M/(N\sigma_0^2)$. 
\end{itemize}

Table \ref{table1} is a summary of the comparison between these two architectures when $\bar{\mH}\bar{\mH}^H = \I_M$ and $\mH\mH^H = (N/M)\I_M$, where $\I_M$ is the identity matrix of dimension $M$. Note that both architectures lead to the same discrete signal model \eqref{mixing}. In the following, we will focus on signal recovery and signature waveform designs based on \eqref{mixing}.

\begin{table}[htp]
\begin{center}
 \caption{Comparison of the two architectures.}
    \begin{tabular}{|c|c|c|}
    \hline
    Architecture & Chip-rate  & Generalized matched  \\ 
    & subsampling & filter bank  \\ \hline
       $\#$ of Users & $ N$ &  $ N$  \\ \hline
         $\#$ of Filters & 1 & $M$  \\ \hline
    $\#$ of Samples &  $M$   & 1 \\ \hline
         Sampling Rate & $(N/M) T_c$ & $ T$ $(T \gg T_c)$ \\ \hline
    Signal Power & $M/N$ & 1 \\ \hline
    Noise Power  & $\sigma_0^2$ & $(N/M)\sigma_0^2$ \\ \hline
     SNR per  & $M/N$ & $M/N$ \\ 
     measurement & & \\ \hline
    \end{tabular}
    \label{table1}
\end{center}
\end{table}

%% file: algorithms.tex
\section{Coherent and Noncoherent Detection Algorithms}\label{sec:algorithms} 

In the following sections, we choose $\mH$ as a tight-frame, and hence the noise is white and we assume the noise variance is $\sigma^2 \triangleq N/M \sigma_0^2$. Define 
\begin{equation}
\X = \mH\A = [\xb_1,\cdots,\xb_{N_\tau}] \in \mathbb{C}^{M\times N_\tau}. 
\end{equation}
We further assume that the columns of $\mH$ and $\A$ are scaled so that each column of $\X$ is unit-norm: $\|\xb_n\|_2 = 1$.
Hence the model (\ref{mixing}) becomes
\begin{equation}
\yb = \X \R \bb + \wb, \label{Xsig_model}
\end{equation}
where $\wb\sim\mathcal{CN}(\vect{0},\sigma^2 \I_M)$. Based on this model, we first present a coherent matching pursuit detector based on iterative thresholding to detect active users and their transmitted symbols, when $\R$ is assumed known. We also present a noncoherent matching pursuit detector to detect active users when $\R$ is assumed unknown, which is adapted from the Orthogonal Matching Pursuit (OMP) algorithm \cite{Tropp_OMP}.

\subsection{Coherent and Noncoherent Matching Pursuit Detector}
The coherent matching pursuit detector is described in Algorithm~\ref{algorithm_main}. With knowledge of the number of active users $K$, the algorithm performs $K$ iterations. In each iteration, Algorithm~\ref{algorithm_main} first finds a user with the strongest correlation with its delayed signature waveforms, then subtracts its exact contribution to the received signal and updates the residual. Since we assume a flat-fading channel\footnote{Our results can be easily generalized to a multipath channel model.}, there is only one nonzero entry in each user block. Therefore, in the next iteration we can restrict our search to the remaining users. To find the transmitted symbols of each active user, we adopt simple quadrant detectors as in \eqref{signdetector} and \eqref{signdetector2}. Our algorithm doubles the rate of the modulation scheme in \cite{XieEldarGoldsmith2011} by considering the complex nature of the power profiles.

\begin{algorithm}[htp]
\caption{Coherent Matching Pursuit Detector for Asynchronous MUD}
\label{algorithm_main}
\begin{algorithmic}[1]
\STATE Input: matrices $\X$ and $\R$, signal vector $\yb$, number of active users $K$
\STATE Output: active user set $\hat{\mathcal{I}}$, transmitted symbols $\hat{\bb}$
\STATE Initialize: $\mathcal{I}_0 :=$ empty set, $\hat{\bb}_0  := \vect{0}$, $\vb_0 := \yb$, $\mathcal{X}_0 :=\{1, \dotsc, N_\tau\}$
\FOR{$k = 0~\mbox{to}~K-1$}
\STATE Compute: $\fb:= \X^H\vb_{k}$
\STATE Find $i = \arg\max_{n\in\mathcal{X}_k} |f_n|$
\STATE Detect active users: $\mathcal{I}_{k+1} = \mathcal{I}_k \cup\{\lceil i/(\tau+1) \rceil\}$
\STATE Update: $\mathcal{X}_{k+1} = \mathcal{X}_k
\backslash \{\lfloor i/(\tau+1) \rfloor (\tau + 1)+1, \cdots, \lceil i/(\tau+1)\rceil(\tau+1)  \}$
\STATE Detect symbols: 
\begin{equation}\label{signdetector}
\Re\{[\hat{\bb}_{k+1}]_i\} = \frac{1}{\sqrt{2}}\sign(\Re[r_i^* f_i]),
\end{equation}
\begin{equation}\label{signdetector2}
\Im\{[\hat{\bb}_{k+1}]_i\} = \frac{1}{\sqrt{2}}\sign(\Im[r_i^* f_i]),
\end{equation} 
where $\sign$ is the sign function, and $\Re(x)$ and $\Im(x)$ takes the real part and imaginary part of $x$ respectively.
\STATE Update $\hat{\bb}$: $[\hat{\bb}_{k+1}]_n = [\hat{\bb}_k]_n$ for $n\neq i$. \label{symbol_det}
\STATE Update residual: $\vb_{k+1} = \vb_{k} - \X \R\bb_{k+1}$
\ENDFOR
\STATE $\hat{\mathcal{I}} = \mathcal{I}_K$, $\hat{\bb} = \hat{\bb}_K$
\end{algorithmic} 
\end{algorithm}

The noncoherent matching pursuit detector is described in Algorithm~\ref{algorithm_noncoherent}. We denote $\X_\mathcal{I}$ the submatrix (subvector) consisting of columns (entries) of $\X$ indexed by $\mathcal{I}$. Given one symbol, it is not possible to resolve the ambiguity in channel phase. Algorithm~\ref{algorithm_noncoherent} detects whether a user is active or inactive, and does not recover the transmitted symbols. The residual is updated by subtracting the orthogonal projection of $\yb$ onto the signal space of the detected users. The noncoherent detector is appropriate for the situation where we do not have access to the channel state information and are only interested in detecting the active users. For example, it is more pertinent in applications like RFID where it is only important to register the presence or absence of a tag. 

\begin{algorithm}[htp]
\caption{Noncoherent Matching Pursuit Detector for Asynchronous MUD}
\label{algorithm_noncoherent}
\begin{algorithmic}[1]
\STATE Input: matrix $\X$, signal vector $\yb$, number of active users $K$
\STATE Output: active user set $\hat{\mathcal{I}}$
\STATE Initialize: $\mathcal{I}_0 :=$ empty set, $\vb_0 := \yb$, $\mathcal{X}_0 :=\{1, \dotsc, N_\tau\}$
\FOR{$k = 0~\mbox{to}~K-1$}
\STATE Compute: $\fb:= \X^H\vb_{k}$
\STATE Find $i = \arg\max_{n\in\mathcal{X}_k} |f_n|$
\STATE Detect active users: $\mathcal{I}_{k+1} = \mathcal{I}_k \cup\{\lceil i/(\tau+1) \rceil\}$
\STATE Update: $\mathcal{X}_{k+1} = \mathcal{X}_k
\backslash \{\lfloor i/(\tau+1) \rfloor (\tau + 1)+1, \cdots, \lceil i/(\tau+1)\rceil(\tau+1)  \}$
\STATE Update residual: $\vb_{k+1} = \yb -\X_{\mathcal{I}_{k+1}}\X_{\mathcal{I}_{k+1}}^{\dag}\yb$.
\ENDFOR
\STATE $\hat{\mathcal{I}} = \mathcal{I}_K$
\end{algorithmic} 
\end{algorithm}

The complexity of the coherent detector is lower than that of the noncoherent detector, since no orthogonalization is necessary to update the residual. In both detectors, it is possible to terminate the algorithm early and obtain partial recovery of active users\footnote{The performance guarantees can be easily generalized to partial recovery.}. 

\subsection{Performance Guarantees}
The performance guarantee for the two algorithms are expressed in terms of two fundamental metrics of coherence of $\X$. The first is the worse-case coherence: 
\begin{equation}
\mu(\X)\triangleq \max_{n\neq m}|\xb_n^H \xb_m|, 
\end{equation}
which is widely used in characterizing the performance of sparse recovery algorithms. The second is the average coherence, defined as
\begin{align}
\nu(\X) & \triangleq \frac{1}{N_\tau-1}\max_{n}\left| \sum_{m\neq n}\xb_n^H\xb_m
\right| ,
\end{align}
where $\vect{1}$ is an all-one vector. 

We say that a matrix $\X$ satisfies the \textit{coherence property} if the following two conditions hold:
\ben \label{CP}
\mu(\X)\leq \frac{0.1}{\sqrt{2\log N_\tau}}, \quad\quad \nu(\X)\leq \frac{\mu(\X)}{\sqrt{M}}. 
\een
In addition, we say that that a matrix $\X$ satisfies the \textit{strong coherence property} if the following two conditions hold:
\begin{equation} \label{SCP}
\mu(\X) \leq \frac{1}{240\log N_\tau}, \quad\quad \nu(\X) \leq \frac{\mu(\X)}{\sqrt{M}}.
\end{equation}

Note that the condition on average coherence $ \nu(\X) \leq \mu(\X)/\sqrt{M}$ can be achieved with essentially no cost via ``wiggling'', i.e. flipping the signs (or phases) of the columns of $\X$, which doesn't change the worst-case coherence $\mu(\X)$ and the spectral norm $\|\X\|_2$ \cite{BCD2011}. For simplicity we shall write $\mu=\mu(\X)$ and $\nu=\nu(\X)$.

We sort the amplitude of the entries of an $K$-sparse vector $\rb$,  $|r_n|$ from the largest to the smallest for the active users and denote as $|r|_{(1)}, \dotsc, |r|_{(K)}$. Let \[|r|_{\min}=|r|_{(K)}.\] We define the $n$th Signal-to-Noise Ratio ($\textsf{SNR}_n$) and the $n$th Largest-to-Average Ratio ($ \LAR_n$) as
\begin{equation*}
\SNR_{n} = \frac{|r|_{(n)}^2}{\mathbb{E}\{\|\wb\|_2^2\}/K},\quad  \LAR_n = \frac{|r|_{(n)}^2}{\|\rb\|_2^2/K}, \quad n = 1, \dotsc, K.
\end{equation*}
The Signal-to-Noise Ratio ($\SNR$) and minimum Signal-to-Noise Ratio ($\SNR_{\min}$) are defined respectively as
\begin{equation*}
\SNR  = \frac{\|\rb\|^2}{\mathbb{E}\|\wb\|_2^2}, \quad \SNR_{\min}  =  \frac{|r|_{\min}^2}{\mathbb{E}\|\wb\|_2^2/K}. 
\end{equation*}

We then have the following performance guarantee for the coherent matching pursuit detector.

\begin{thm}
\label{thm-coherent}
Suppose that $N_\tau=N(\tau+1)\geq 128$, that the noise $\wb$ is distributed as $\mathcal{CN}(\mathbf{0},\sigma^2\I_M)$, and that $\X$ satisfies the coherence property. If the number of active users satisfies 
\begin{equation} \label{k-cond-coherent}
K\leq \min\left\{  \frac{M}{2\log N_\tau}, \frac{1}{c^2\mu^2\log N_\tau} \right\}
\end{equation}
for $c=20\sqrt{2}$, and if the power profile of active users satisfies
\begin{equation} \label{lar-condition-coherent}
 \LAR_{(k)} >\frac{8}{(1-c\mu\sqrt{(K-k+1) \log N_\tau})^2 } \cdot \left( \frac{K \log N_\tau}{M \SNR } \right),  
 \end{equation}
for $1\leq k\leq K$, then Algorithm~\ref{algorithm_main}  satisfies 
$$\Pr\{\hat{\bb}\neq{\bb}\} \leq (4+\pi^{-1})N_\tau^{-1}.$$ 
\end{thm}

Since $\LAR_{(k)}\geq \LAR_{(K)}$ for $1\leq k\leq K$, \eqref{lar-condition-coherent} can be satisfied if
\begin{equation} \label{mar-condition}
 \LAR_{(K)} >\frac{8}{(1-c\mu\sqrt{\log N_\tau})^2 } \cdot \left( \frac{K \log N_\tau}{M \SNR } \right).  
 \end{equation}
Let $\theta=c\mu\sqrt{K\log N_\tau}\in(0,1)$, then \eqref{mar-condition} implies that the number of active users is bounded by
\begin{equation*} 
K<\frac{M(1-\theta)^2\SNR_{\min}}{8\log N_\tau}. 
\end{equation*}
Combining this with \eqref{k-cond-coherent}, we have the following corollary.

\begin{corollary} \label{thm-coherent1}
Suppose that $N_\tau=N(\tau+1)\geq 128$, that the noise $\wb$ is distributed as $\mathcal{CN}(\mathbf{0},\sigma^2\I_M)$, and that $\X$ satisfies the coherence property. We write $\mu = c_1 M^{-1/\gamma}$ for some $c_1 > 0$ ($c_1$ may depend on $N_\tau$ and $\gamma \in \{0\}\cup [2, \infty)$). Then Algorithm \ref{algorithm_main} satisfies $\Pr\{\hat{\bb}\neq\bb \} \leq (4+\pi^{-1})N_\tau^{-1}$ as long as the number of active users $K$ satisfies 
\begin{align} \label{k-all-cond}
K <\max_{0<\theta<1} \min\left\{  \frac{M}{2\log N_\tau}, \frac{M(1-\theta)^2\mathsf{SNR}_{\min}}{8\log N_\tau}, \frac{\theta^2M^{2/\gamma}}{c_2^2\log N_\tau} \right\},
\end{align}
where $c_2=20\sqrt{2}c_1$. 
\end{corollary}

We have the following performance guarantee for the noncoherent matching pursuit detector.

\begin{thm} \label{thm-noncoherent}
Suppose that $N_\tau=N(\tau+1)\geq 128$, that the noise $\wb$ is distributed as $\mathcal{CN}(\mathbf{0},\sigma^2\I_M)$, and that $\X$ satisfies the coherence property. If the number of active users satisfies 
\begin{equation} \label{k-cond}
K\leq \min\left\{ \frac{N_\tau}{c_4^2\|\X\|_2^2\log N_\tau}, \frac{1}{c_3^2\mu^2\log N_\tau} \right\}
\end{equation}
for $c_3=50\sqrt{2}$ and $c_4=104\sqrt{2}$, and if the power profile of active users satisfies
\begin{equation} \label{lar-condition}
 \LAR_{(k)} >\frac{8}{(1-c_3\mu\sqrt{(K-k+1) \log N_\tau})^2 } \cdot \left( \frac{K \log N_\tau}{M \SNR } \right),  
 \end{equation}
for $1\leq k\leq K$, then Algorithm~\ref{algorithm_noncoherent}  satisfies 
$$\Pr\{\hat{\mathcal{I}}\neq\mathcal{I}\} \leq (K\pi^{-1}+6)N_\tau^{-1}.$$ 
\end{thm}

Similarly, using the fact that $\LAR_{(k)}\geq \LAR_{(K)}$ for $1\leq k\leq K$, \eqref{lar-condition} can be satisfied if
\begin{equation} \label{mar-condition-noncoherent}
 \LAR_{(K)} >\frac{8}{(1-c_3\mu\sqrt{\log N_\tau})^2 } \cdot \left( \frac{K \log N_\tau}{M \SNR } \right),  
 \end{equation}
and the following corollary becomes straightforward.

\begin{corollary} \label{thm-noncoherent1}
Suppose that $N_\tau=N(\tau+1)\geq 128$, that the noise $\wb$ is distributed as $\mathcal{CN}(\mathbf{0},\sigma^2\I_M)$, and that $\X$ satisfies the coherence property. We write $\mu = c_1 M^{-1/\gamma}$ for some $c_1 > 0$ ($c_1$ may depend on $N_\tau$ and $\gamma \in \{0\}\cup [2, \infty)$). Then Algorithm \ref{algorithm_noncoherent} satisfies $\Pr\{\hat{\mathcal{I}}\neq\mathcal{I}\} \leq (K\pi^{-1}+6)N_\tau^{-1}$ as long as the number of active users $K$ satisfies
\begin{align} 
K <\max_{0<\theta<1} \min  \Big\{& \frac{M(1-\theta)^2\SNR_{\min}}{8\log N_\tau}, \frac{\theta^2M^{2/\gamma}}{c_3^2\log N_\tau}, \nonumber \\
&\quad \frac{N_\tau}{c_4^2\| \X\|_2^2\log N_\tau} \Big\},\label{k-all-noncoherent}
\end{align}
where $c_3=50\sqrt{2}c_1$ and $c_4=104\sqrt{2}$. 
\end{corollary}

Theorem~\ref{thm-coherent} and Theorem~\ref{thm-noncoherent} implies that with both coherent and noncoherent detectors, the system can support $K\sim\mathcal{O}(M/\log N_\tau)$ users with $M$ samples. In other words, the system can support $K$ users with $M\sim\mathcal{O}(K\log N_\tau)$ samples. Since both detectors are based on iterative thresholding, the power profile of different active users, defined in \eqref{powerprofile}, enters the analysis only through the quantities $\LAR_{(k)}$, and plays a less important role than $\SNR_{\min}$, the $\SNR$ of the weakest active user in determining the performance. Performance of the two algorithms is identical in terms of scaling.

%% file: proofs.tex
\section{Proofs of Main Theorems}\label{sec:proofs} 
Central to the proof is the notion of $(K, \epsilon, \delta)$-Statistical Orthogonality Condition (StOC) introduced in \cite{BajwaCalderbankJafarpour2010}, which can be related to the worst-case and average coherence of matrix $\X$. We prove that the probability of error is vanishingly small if with high probability $\X$ satisfies the StOC and the noise $\wb$ is uniformly bounded.
\subsection{Preparations}
We first introduce an alternative way to represent the measurement model. We can write the vector of transmitted symbols together with the power $\R\bb$ as a concatenation of a random permutation matrix and a deterministic $K$-sparse vector $\bar{\zb}\in\mathbb{C}^{N_\tau}$. The form of the $K$-sparse vector is given by $\bar{\zb} \triangleq [z_1, \cdots, z_k, 0, \cdots, 0]\transpose$. Let $\bar{\Pi}\triangleq (\pi_1, \cdots, \pi_{N_\tau})$ be a random permutation of $\lsem N_\tau \rsem$. Let $\vect{P}_{\bar{\Pi}}$ be a $N_\tau\times N_\tau$ permutation matrix, and $\vect{P}_{\bar{\Pi}} \triangleq [\eb_{\pi_1}, \cdots, \eb_{\pi_{N_{\tau}}}]\transpose$, with $\eb_{n}$ being the $n$th column of the identity matrix $\I_{N_\tau}$. Given this notation, the assumption that  $\mathcal{I}$ is a random subset of $\lsem N_\tau \rsem $ is equivalent to stating that $\bar{\zb} = \Pb_{\bar{\Pi}}\R\bb$. Hence the measurement equation (\ref{Xsig_model}) can be written as
\ben
\yb = \X\R\bb + \wb = \X\vect{P}_{\bar{\Pi}}\bar{\zb} + \wb = \X_{\Pi}\rb_{\mathcal{I}} + \wb,
\een
where $\Pi\triangleq (\pi_1, \cdots, \pi_k)$ denotes the first $k$ elements of the random permutation $\bar{\Pi}$, and $\X_{\Pi}$ denotes the $M\times K$ sub-matrix obtained by collecting the columns of $\X$ corresponding to the indices in $\Pi$, and the vector $\rb_{\mathcal{I}}\in\mathbb{C}^K$ represents the $K$ nonzero entries of $\R\bb$. We next define the \textit{$(K, \epsilon, \delta)$-Statistical Orthogonality Condition (StOC)}. 
\begin{mydef}[StOC]
Let $\bar{\Pi}$ be a random permutation of $\lsem N_\tau \rsem$. Define $\Pi \triangleq (\pi_1, \cdots, \pi_K)$ and $\Pi_c \triangleq (\pi_{K+1}, \cdots, \pi_{N_\tau})$ for any $K\in [1, N_\tau]$. Then, the $M\times N_\tau$ (normalized) matrix $\X$ is said to satisfy the $(K, \epsilon, \delta)$-statistical orthogonality condition if there exists $\epsilon, \delta \in [0, 1)$ such that the inequalities:
\begin{align}
\|(\X^H_\Pi \X_\Pi - \I_K)\zb\|_\infty & \leq \epsilon \|\zb\|_2 \quad\mbox{\quad(StOC-1)}  
\label{stoc1} \\
 \|\X^H_{\Pi_c} \X_\Pi \zb\|_\infty & \leq \epsilon \|\zb\|_2  \quad\mbox{\quad(StOC-2) } \label{stoc2}
\end{align}
hold for every fixed $\zb\in\mathbb{C}^K$ with probability exceeding $1-\delta$ with respect to the random permutation $\bar{\Pi}$.  
\end{mydef}
The StoC property has proved useful in obtaining average case performance guarantees \cite{BajwaCalderbankJafarpour2010,allerton_chi}. It is similar in spirit to the Restricted Isometry Property (RIP) \cite{CandesTao2006} which provides worst case guarantees in CS. An important difference between the two is that while we know of no effective algorithm for testing RIP, it is possible to infer StOC from matrix invariants that can be easily computed.

If \eqref{stoc1} and \eqref{stoc2} hold for a realization of permutation $\bar{\Pi}$, then for $1\leq k< K$, let $\Pi_k=(\pi_1,\ldots,\pi_{k})$ and $\Pi_t^c=(\pi_{k+1},\ldots,\pi_{K})$, so that $\Pi_t\cup\Pi_t^c=\Pi$ and $\Pi_t\cap\Pi_t^c=\emptyset$. For every $\zb \in\mathbb{C}^{k}$, we have from \eqref{stoc1} that
\begin{align*}
\left \| \begin{bmatrix}
 \X_{\Pi_k}^H\X_{\Pi_k}-\I_k  & \X_{\Pi_k}^H \X_{\Pi_k^c} \\
  \X_{\Pi_k^c}^H \X_{\Pi_k} &  \X_{\Pi_k^c}^H\X_{\Pi_k^c}-\I_{K-k}
  \end{bmatrix} \begin{bmatrix}
  \zb \\
  \mathbf{0}
  \end{bmatrix} \right\|_{\infty} & \leq\epsilon \|\zb\|_2.
\end{align*}
Therefore $\| (\X_{\Pi_k}^H\X_{\Pi_k}-\I_k)\zb \|_{\infty} \leq\epsilon \|\zb\|_2$, and
$\| \X_{\Pi_k^c}^H\X_{\Pi_k} \zb \|_{\infty} \leq\epsilon \|\zb \|_2$. Moreover, from \eqref{stoc2} we have
\begin{align*}
\| \X_{\Pi_k^c}^H\X_{\Pi_k} \zb \|_{\infty} & = \left\| \begin{bmatrix}
  \X_{\Pi_k^c}^H \X_{\Pi_k} &  \X_{\Pi_k^c}^H\X_{\Pi_k^c}
  \end{bmatrix} \begin{bmatrix}
  \zb \\
  \mathbf{0}
  \end{bmatrix} \right\|_{\infty} \leq\epsilon \|\zb\|_2.
\end{align*}

We  also need the following two lemmas.

\begin{lemma} \label{lemma-stoc}
An $M\times N_\tau$ matrix $\X$ satisfies $(K, \epsilon, \delta)$-StOC for any $\epsilon\in[0, 1)$ and $a \geq 1$ with 
\ben
\delta \leq 4 N_\tau \exp\left(
-\frac{(\epsilon - \sqrt{k}\nu)^2}{16(2+a^{-1})^2\mu^2}
\right),
\een
as long as $K \leq \min\{\epsilon^2 \nu^{-2}, (1+a)^{-1} N_\tau\}$.
\end{lemma}
The proof for this lemma can be found in \cite{BajwaCalderbankJafarpour2010}. A consequence of this lemma is that if we let $K \leq M/(2\log N_\tau)$ and fix $\epsilon = 10\mu\sqrt{2\log N_\tau}$, then the matrix $\X$ satisfies $(K, \epsilon, \delta)$-StOC with $\delta \leq 4N_\tau^{-1}$. Define the event $\mathcal{G}_1$ as:
\begin{equation}
\mathcal{G}_1\triangleq \{\X \mbox{ satisfies StOC-1 and StOC-2} \}. \label{g_def}
\end{equation}
Then $\cG_1$ occurs with probability at least $1-4N_\tau^{-1}$ with respect to $\bar{\Pi}$ given the aforementioned choice of parameters.

In order to prove Theorem~\ref{thm-noncoherent}, we need an argument due to Tropp \cite{Tropp} that shows a random submatrix of $\X$ is well-conditioned with high probability. We follow the treatment given by Cand\`es and Plan  \cite{CandesPlan} where this argument appears in a slightly different form, given below.

\begin{lemma}[\cite{Tropp,CandesPlan}] \label{submatrix}
Let $\bar{\Pi}=(\pi_1,\ldots,\pi_{N_{\tau}})$ be a  random permutation of $\lsem N_\tau \rsem$, and define $\Pi=(\pi_1,\ldots,\pi_K)$ for any $K\in [1, N_\tau]$. Then for $q=2\log N_\tau$ and $K\leq N_\tau/(4\|\X\|_2^2)$, we have
\begin{align}
& \quad \left(\mathbb{E}\left[ \|\X_{\Pi}^H\X_{\Pi} - \I_K \|_2^q \right] \right)^{1/q}  \nonumber \\ &\leq 2^{1/q} \left(30\mu\log N_\tau +13\sqrt{\frac{2K\|\X\|_2^2\log N_\tau}{N_\tau}} \right).
\end{align}
with respect to the random permutation $\bar{\Pi}$. 
\end{lemma}

The following lemma \cite{CandesPlan} states a probabilistic bound on the extreme singular values of a random submatrix of $\X$, by applying the Markov inequality to Lemma~\ref{submatrix}:
$$ \Pr\left( \|\X_{\Pi}^H\X_{\Pi}-\I_K \|_2\geq 1/2\right) \leq 2^q \mathbb{E}\left[ \| \X_{\Pi}^H\X_{\Pi} - \I_K \|_2^q\right] $$
\begin{lemma}[\cite{CandesPlan}] \label{submat}
Let $\bar{\Pi}=(\pi_1,\ldots,\pi_{N_\tau})$ be a  random permutation of $\lsem N_\tau \rsem$, and define $\Pi=(\pi_1,\ldots,\pi_K)$ for any $K\leq N_\tau$. Suppose that $\mu\leq 1/(240\log N_\tau)$ and $K\leq N_\tau/(c_2^2\|\X\|_2^2\log N_\tau)$ for numerical constant $c_2=104\sqrt{2}$, then we have
$$ \Pr\left( \|\X_{\Pi}^H\X_{\Pi}-\I_K \|_2\geq \frac{1}{2} \right) \leq 2p^{-2\log 2}.$$
\end{lemma}
Define the event \[\cG_2 \triangleq\{  \| \X_{\Pi}^H\X_{\Pi}-\I_K \|_2\leq 1/2 \},\] which happens at least $1-2N_\tau^{-2\log 2}> 1-2N_\tau^{-1}$ with respect to $\bar{\Pi}$ from Lemma~\ref{submat}. Notice that 
all the eigenvalues of $\X_{\Pi}^H\X_{\Pi}$ are bounded in $[1/2, 3/2]$. Under $\cG_2$, we have $\| (\X_{\Pi}^H\X_{\Pi})^{-1}\|_2\leq 2$ and $\|\X_{\Pi}(\X_{\Pi}^H\X_{\Pi})^{-1}\|_2\leq\sqrt{2}$.
Moreover, for $1\leq k < K$ and $\Pi_k=(\pi_1,\ldots,\pi_k)$, we have $\|\X_{\Pi_k}^H\X_{\Pi_k}-\I_k \|_2 \leq 1/2$, since eigenvalues of $\X_{\Pi_k}^H\X_{\Pi_k}$ are majorized by eigenvalues of $\X_{\Pi}^H\X_{\Pi}$ \cite{marshall2010inequalities}.

Finally, we need that the noise is bounded with high probability. 

\begin{lemma} \label{noise}
Let $\Pb\in\mathbb{C}^{M\times M}$ be a projection matrix such that $\Pb^2=\Pb$.
Let $\wb \sim \mathcal{CN}(\vect{0}, \sigma^2 \I)$, and $\X$ be a unit-column matrix. Then for $\tau>0$ we have
\begin{align*}
\Pr(\| \X^H\Pb\wb \|_\infty \leq \tau)& \geq 1-\frac{N_\tau}{\pi}e^{-\tau^2/\sigma^2},
\end{align*}
provided the right hand side is greater than zero.
\end{lemma}
\begin{proof}
See Appendix~\ref{proof_noise}.
\end{proof}
Now let $\tau=\sigma\sqrt{ 2\log N_\tau}$, and define
$$\mathcal{H}_0 = \{ \|\X^H\Pb\wb \|_\infty \leq\tau \}.$$
It follows from Lemma \ref{noise} that $\mathcal{H}_0$ occurs with probability at least $1-\pi^{-1}N_\tau^{-1}$.

\subsection{Proof of Theorem~\ref{thm-coherent}}
When applied to Algorithm~\ref{algorithm_main}, the next lemma shows that under appropriate conditions, ranking the inner products between $\xb_n$ and $\yb$ is an effective method of detecting the set of active users. %
\begin{lemma}\label{lemma_active_user}
Let $\bb$ be a vector with support $\mathcal{I}$ corresponding to $K$ active users, and let $\yb$ be a noisy measurement as in \eqref{Xsig_model}. Suppose that
\begin{equation}
|r|_{(1)} - 2 \epsilon  \|\rb_{\mathcal{I}}\|_2  >  2\tau. \label{cond_2}
\end{equation}
Then, if the event $\cG_1\cap\cH_0$ occurs, we have
\begin{equation}
\max_{n\in\mathcal{I}} |\xb_n^H \yb| > \max_{n\notin\mathcal{I}} |\xb_n^H \yb|.\label{rank_OMP}
\end{equation}
and $\sign( \Re[r_{n_1}^*\xb_{n_1}^H\yb])=\sqrt{2}  \Re[b_{n_1}],$ $\sign( \Im[r_{n_1}^*\xb_{n_1}^H\yb])=\sqrt{2}  \Im[b_{n_1}],$ 
for 
\begin{equation}
n_1 = \arg\max_n |\xb_n^H \yb|. \label{n_1}
\end{equation}
\end{lemma}

\begin{proof} See Appendix~\ref{proof_active_user}.
\end{proof}


We now prove the performance guarantee for the coherent detector in Algorithm~\ref{algorithm_main}. First we show that under the event $\mathcal{G}_1\cap\mathcal{H}_0$, which happens with probability at least $1-(\pi^{-1}+4)N_\tau^{-1}$, Algorithm \ref{algorithm_main} correctly detects all active users and symbols in the first $K$ iterations. Define a subset $\Pi_k$ which contains the $k$ variables that are selected until the $k$th iteration, $0\leq k\leq K$. 

We want to prove $\Pi_k\subset\Pi$ by induction. First at $k=0$, $\Pi_0=\emptyset\subset\Pi$. Suppose we are currently at the $k$th iteration of Algorithm~\ref{algorithm_main}, $0\leq k \leq K-1$, and assume that $\Pi_k\subset\Pi$. The $k$th step is to detect the user with the largest $|\xb_n^H \vb_{k}|$. We have
\begin{equation} 
\vb_k = \X (\bb - \bb_{k}) + \wb = \X{\bm\eta}_{k} + \wb, \label{model_2}
\end{equation}
where ${\bm\eta}_{k}\triangleq \bb - \bb_{k}$. This vector has support $\Pi_k^c = \Pi\backslash\Pi_{k}$ and has at most $(K-k)$ non-zero elements, since $\bb_{k-1}$ contains correct symbols at the correct locations for $k$ active users, i.e. $[\bb_{k}]_n = [\bb]_n$, for $ n\in\Pi_k$. This $\vb_k$ is a noisy measurement of the vector $ \X {\bm\eta}_k$. The signal model in \eqref{model_2} for the $k$th iteration is identical to the signal model in the first iteration with $\bb$ replaced by ${\bm\eta}_{k}$ (with a smaller sparsity $K-k$ rather than $K$), $\Pi$ replaced by $\Pi_k^c$, and $\yb$ replaced by $\vb_{k}$.
Hence, from Lemma~\ref{lemma_active_user} we have that under the condition 
\begin{equation} \label{r-coherent}
\|\rb_{\cI_k^c}\|_{
\infty} - 2 \epsilon  \|\rb_{\mathcal{I}_k^c}\|_2  >  2\tau, 
\end{equation}
we have
\begin{equation}
\max_{n \in \cI\backslash\cI_{(k-1)}} |\xb_n^H  \vb_k| > \max_{n \notin \cI\backslash\cI_{(k-1)}} |\xb_n^H \vb_k |. \label{41}
\end{equation}
i.e. Algorithm~\ref{algorithm_main} can detect an active user correctly, and no index of an active user that has been detected before will be chosen again. Note that $\|\rb_{\mathcal{I}_k^c}\|_\infty \geq |r|_{(k+1)}$, $ \| \rb_{\mathcal{I}_k^c} \|_2 \leq \sqrt{K-k}|r|_{(k+1)}$, \eqref{r-coherent} is satisfied by 
$$  |r|_{(k+1)} > 2\epsilon \sqrt{K-k}  |r|_{(k+1)} + 2\tau. $$
Since $K<1/(c^2\mu^2\log N_\tau)$ and $\epsilon=10\mu\sqrt{2\log N_\tau}$, this is equivalent to the condition in \eqref{lar-condition-coherent} for $0\leq k\leq K-1$, therefore a correct user is selected at the $k$th iteration, so that $\Pi_{k+1}\subset\Pi$. On the other hand, since condition (\ref{cond_2}) is true, the symbol can be detected correctly as well. Then we have that under the event $\mathcal{G}_1\bigcap\mathcal{H}_0$,
$\sign( \Re[r_{n_1}^*\xb_{n_1}^H\yb])=\sqrt{2}  \Re[b_{n_1}]$, $\sign( \Im[r_{n_1}^*\xb_{n_1}^H\yb])=\sqrt{2}  \Im[b_{n_1}]$,
that is $\mathcal{G}_1\bigcap\mathcal{H}_0\subset \{{b}_{n_k}^{(k)} = b_{n_k}\}$. By induction, since no active users will be detected twice, it follows that the first $K$ steps of Algorithm~\ref{algorithm_main} can detect all active users.

\subsection{Proof of Theorem \ref{thm-noncoherent}}
We note that in Algorithm~\ref{algorithm_noncoherent}, the residual $\vb_k$, $k=0,\cdots, K-1$ is orthogonal to the selected columns in previous iterations, so in each iteration a new column will be selected. Define a subset $\Pi_k$ which contains the $k$ variables that are selected until the $k$th iteration. Then $\Pb_k=\X_{\Pi_k}(\X_{\Pi_k}^H\X_{\Pi_k})^{-1}\X_{\Pi_k}^H$ is the projection matrix onto the linear subspace spanned by the columns of $\X_{\Pi_k}$, and we assume $\Pb_0=\mathbf{0}$. 

Again we want to prove $\Pi_k\subset\Pi$ by induction. First at $k=0$, $\Pi_0=\emptyset\subset\Pi$. Assume at the $k$th iteration, $\Pi_k\subset\Pi$, $0\leq k\leq K-1$, then the residual $\vb_k$ can be written as
\begin{align*}
 \vb_k &= (\I-\Pb_k)\yb \\
 & = (\I-\Pb_k)\X_{\Pi}\rb_{\cI} +(\I-\Pb_k)\wb \triangleq \s_k+\n_k,
 \end{align*}
where $\s_k = (\I-\Pb_k)\X_{\Pi}\rb_{\cI}$ and $\n_k = (\I-\Pb_k)\wb$ are the signal and noise components respectively at the $k$th iteration.
 
Let $M_{\Pi}^k = \|\X_{\Pi}^H \s_k\|_\infty$, $M_{\Pi^c}^k = \|\X_{\Pi^c}^H \s_k\|_\infty$ and $N^k =\|\X^H \n_k \|_{\infty}$, then a sufficient condition for $\Pi_{k+1}\subset\Pi$, i.e. for Algorithm~\ref{algorithm_noncoherent} to select a correct active user at the next iteration is that 
\begin{equation} \label{cond}
M_{\Pi}^k-M_{\Pi^c}^k>2N^k
\end{equation}
since under (\ref{cond}) we have
\[  \|\X_{\Pi}^H \vb_k\|_\infty   \geq M_{\Pi}^k -N^K > M_{\Pi^c}^k+N^K \geq  \|\X_{\Pi^c}^H \vb_k\|_\infty.\]


Let the event \[\mathcal{G}=\cG_1\cap \cG_2.\] From Lemma~\ref{lemma-stoc}  and Lemma~\ref{submat} the event $\mathcal{G}$ holds with probability at least $1-4N_{\tau}^{-1}-2N_{\tau}^{-2\log 2}$ with respect to $\bar{\Pi}$. 

Now we bound $M_{\Pi}^k$ and $M_{\Pi^c}^k$ under the event $\cG$. Let $\Pi_k^c = \Pi\backslash\Pi_k$ be the index set of yet to be selected active users, and $\rb_{\Pi_k^c}=\rb_{\cI_k^c}$ be the corresponding coefficients. We can find a vector $\zb$ of dimension $(K-k)$ such that $\X_{\Pi_k^c}^H\X_{\Pi_k^c}\zb=\X_{\Pi_k^c}^H(\I-\Pb_k) \X_{\Pi}\rb_{\cI}$, where the vector $\zb$ can be written as 
\begin{align*}
\zb & =(\X_{\Pi_k^c}^H\X_{\Pi_k^c})^{-1}\X_{\Pi_k^c}^H(\I-\Pb_k)\X_{\Pi_k^c} \rb_{\cI_k^c} \\
&= \rb_{\cI_k^c} - (\X_{\Pi_k^c}^H\X_{\Pi_k^c})^{-1}\X_{\Pi_k^c}^H\Pb_k\X_{\Pi_k^c} \rb_{\cI_k^c} .
\end{align*}
Since we have
\begin{align}
\|\zb\|_2 & \leq \|(\X_{\Pi_k^c}^H\X_{\Pi_k^c})^{-1}\|_2\|\X_{\Pi_k^c}^H(\I-\Pb_k)\X_{\Pi_k^c} \rb_{\cI_k^c}\|_2 \label{w1} \\
& \leq 2\|\X_{\Pi}^H\X_{\Pi} \|_2 \|\rb_{\cI_k^c} \|_2 \leq 3 \| \rb_{\cI_k^c} \|_2, \label{w2}
\end{align}
where \eqref{w1} follows from
$$\|\X_{\Pi_k^c}^H(\I-\Pb_k)\X_{\Pi_k^c})\|_2 \leq \|\X_{\Pi}^H\X_{\Pi}\|_2, $$
whose proof can be found in \cite{CaiWang}, and \eqref{w2} follows from Lemma~\ref{submat}. Also, 
\begin{align*}
 &\quad \|\X_{\Pi_k^c}^H\Pb_k\X_{\Pi_k^c} \rb_{\cI_k^c}\|_{\infty} \nonumber \\
 &= \|\X_{\Pi_k^c}^H\X_{\Pi_k}(\X_{\Pi_k}^H\X_{\Pi_k})^{-1}\X_{\Pi_k}^H\X_{\Pi_k^c} \rb_{\cI_k^c}\|_{\infty}  \\
 & \leq \epsilon\|(\X_{\Pi_k}^H\X_{\Pi_k})^{-1}\X_{\Pi_k}^H\X_{\Pi_k^c} \rb_{\cI_k^c}\|_2 \\
& \leq \epsilon\|(\X_{\Pi_k}^H\X_{\Pi_k})^{-1} \|_2 \|\X_{\Pi_k}^H\X_{\Pi_k^c}\|_2 \|\rb_{\cI_k^c}\|_2\\
& \leq \epsilon\|\rb_{\cI_k^c}\|_2,
 \end{align*}
 therefore $M_{\Pi}^k$ can be bounded as
\begin{align}
&\quad M_{\Pi}^k = \| \X_{\Pi_k^c}^H\X_{\Pi_k^c} \rb_{\cI_k^c} - \X_{\Pi_k^c}^H\vect{P}_k\X_{\Pi_k^c} \rb_{\cI_k^c} \|_\infty \nonumber \\
& \geq \| \rb_{\cI_k^c} \|_\infty -\|( \X_{\Pi_k^c}^H\X_{\Pi_k^c} -\I) \rb_{\cI_k^c} \|_\infty- \|\X_{\Pi_k^c}^H\Pb_k \X_{\Pi_k^c} \rb_{\cI_k^c} \|_\infty  \nonumber \\
& \geq \|\rb_{\cI_k^c} \|_{\infty}- 2\epsilon \| \rb_{\cI_k^c} \|_2. \label{M1}
\end{align}
where \eqref{M1} follows from \eqref{stoc1}. Next, $M_{\Pi^c}^k$ can be bounded as
\begin{align} 
M_{\Pi^c}^k &= \|\X_{\Pi^c}^H (\I-\Pb_k) \X_{\Pi} \rb_{\cI}\|_\infty \nonumber \\
& = \|\X_{\Pi^c}^H \X_{\Pi_k^c} \zb \|_\infty \nonumber \\
& \leq \epsilon \|\zb \|_2\leq 3\epsilon \| \rb_{\cI_k^c} \|_2. \label{M2}
\end{align}
where \eqref{M2} follows from \eqref{stoc2}. 

Conditioned on the event $\cG$, for each $\Pb_k$, since $\I-\Pb_k$ is also a projection matrix, define the event 
\begin{equation}
\mathcal{H}_k = \{N^k\leq\tau \},\quad k=0,\cdots, K-1. \label{Nt}
\end{equation}
Then from Lemma~\ref{noise}, $\cH_k$ happens with probability at least $1- \pi^{-1}N_\tau^{-1}$ with respect to $\wb$. We further define the event $\mathcal{H}=\cap_{k=0}^{K-1} \mathcal{H}_k$, then from the union bound $\Pr(\cH|\cG) = \Pr(\cH)  \geq 1- K\pi^{-1}N_\tau^{-1}$.

Under the event $\cG\cap\cH$, from the above discussions which happens with probability $\Pr(\cG\cap\cH)\geq 1-K\pi^{-1}N_\tau^{-1}-2N_\tau^{-2\log 2}-4N_\tau^{-1}\geq 1- (K\pi^{-1}+6)N_\tau^{-1}$. Now we are ready to analyze the performance of Algorithm~\ref{algorithm_noncoherent} under the event $\cG\cap\cH$. Substituting the bounds \eqref{M1}, \eqref{M2} and \eqref{Nt} into \eqref{cond}, it is sufficient that at the $k$th iteration
\begin{equation} \label{bound1}
\|\rb_{\mathcal{I}_k^c} \|_{\infty} > 5\epsilon \| \rb_{\mathcal{I}_k^c} \|_2  + 2\tau.
\end{equation}

Note that $\|\rb_{\mathcal{I}_k^c}\|_\infty \geq |r|_{(k+1)}$, $ \| \rb_{\mathcal{I}_k^c} \|_2 \leq \sqrt{K-k}|r|_{(k+1)}$, \eqref{bound1} is satisfied by 
$$  |r|_{(k+1)} > 5\epsilon \sqrt{K-k}  |r|_{(k+1)} + 2\tau. $$
Since $K<1/(c_1^2\mu^2\log N_\tau)$ and $\epsilon=10\mu\sqrt{2\log N_\tau}$, this is equivalent to the condition in \eqref{lar-condition} for $0\leq k\leq K-1$, therefore a correct user is selected at the $k$th iteration, so that $\Pi_{k+1}\subset\Pi$. Since the number of active users is $K$, Algorithm~\ref{algorithm_noncoherent} successfully finds $\Pi$ in $K$ iterations under the event $\cG\cap\cH$, and we have proved Theorem~\ref{thm-noncoherent}.

%% file: waveforms.tex
\section{Deterministic Signature Waveforms}\label{sec:waveforms}

\subsection{Gabor Frames}

In the following we will construct the signature sequences $\tilde{\ab}_n$ from Gabor frames. Let $\g\in\mathbb{C}^P$ be a seed vector with each entry $|g_n|^2=1/M$ and let $\T(\g)\in\mathbb{C}^{P\times P}$ be the circulant matrix generated from $\g$ as $\T(\g) = [\mathcal{T}_0\g\quad \cdots\quad \mathcal{T}_\tau\g]$. 
Its eigen-decomposition can be written as 
$$\T(\g)=\F\mbox{diag}(\F^H\g)\F^H \triangleq \F\mbox{diag}(\hat{\g})\F^H,$$ where $\F=\frac{1}{\sqrt{P}}[\bomega_0, \bomega_1,\cdots,\bomega_{P-1}]$ is the DFT matrix with columns
$$\bomega_m=[e^{j2\pi\frac{m}{P}\cdot 0},e^{j2\pi\frac{m}{P}\cdot 1},\ldots, e^{j2\pi\frac{m}{P}\cdot(P-1)}]\transpose. $$
Define corresponding diagonal matrices $\W_m=\mbox{diag}[\bomega_m]$, for $m=0,1,\ldots, P-1$. Then the Gabor frame $\mPhi=[\mphi_m]$
generated from $\g$ is an $P\times P^2$ block matrix of the form
\begin{equation}\label{gabor}
\mPhi = [\W_0\T(\g) , \; \W_1\T(\g),\; \ldots,\; \W_{P-1}\T(\g)].
\end{equation}
where each column has norm $\sqrt{P/M}$. When we apply the DFT to the Gabor frame $\mPhi$, and obtain $\hat{\mPhi}=\F^H\mPhi$, the order of time-shift and frequency modulation is reversed, and therefore $\hat{\mPhi}$ is composed of circulant matrices after appropriate ordering of columns. In fact, if we index each column $m$ from $P^2$ to $P\times P$ by $m=Pq+\ell$, the matrix $\mPhi_{\ell}$ is obtained by keeping all columns with $r=\ell$ (mod $P$). So $\mPhi_{\ell}$ can be written as
$$\mPhi_{\ell} = \sqrt{P} \cdot \mbox{diag}(\shiftR^\ell \g)\F,$$
where $\shiftR$ is the right-shift matrix by one, and its DFT transform
$$\hat{\mPhi}_{\ell}=\F\mPhi_{\ell}=\sqrt{P}\T(\W_\ell \hat{\g})$$ 
is a circulant matrix. We use $[\mPhi_1, \cdots, \mPhi_{P-1}]$ {as} the matrix $\A$.

At the receiver, a random partial DFT is applied to the received symbol, so $\mH = \F_{\Omega}$ is a partial DFT matrix, and the resulted matrix $\X=\mPhi_{\Omega}$ is a subsampled Gabor frame defined in \eqref{gabor}, {with unit-norm columns}. The maximum discrete delay $\tau$ which this Gabor frame construction can support is ${P}-1$, where $\W_\ell \hat{\g}$ can be assigned as signature sequences to a user, so the maximum number of total users should satisfy $N\leq P$. In general, if $\tau< {P}-1$, we can split $\mPhi_{\ell}$ into blocks to support multiple users, and send $\mathcal{T}_{d(\tau+1)}\W_\ell \hat{\g}$ as signature sequences for $d=0,\cdots,\lfloor P/(\tau+1) \rfloor$ and $\ell=1,\cdots, P$,  so the maximum number of total user satisfies {$N\leq P\lfloor P/(\tau+1) \rfloor$} in general.

Now we consider the coherence properties of $\X$. We have the following proposition. 
\begin{prop} \label{subsample_coherence}
Let $\X$ be a unit-column matrix with $M$ rows subsampled uniformly at random from a Gabor frame $\mPhi$ that satisfies the strong coherence property. If $M\geq \gamma \log^3 P$ for some constant $\gamma$, then with probability at least $1-2P^{-1}$, we have $\mu(\X)\leq \gamma_2/\log P$ for some constant $\gamma_2$, and $\nu(\X)\leq 2/M$ deterministically.
\end{prop}
\begin{proof} See Appendix~\ref{proof_coherence}.
\end{proof}

It is established in \cite{BajwaCalderbankJafarpour2010} that Gabor frames satisfy the strong coherence property when the seed sequence is the Alltop sequence, or with high probability when the seed sequence is randomly generated. Proposition~\ref{subsample_coherence} implies that we can find an $M$ such that the subsampled Gabor frame satisfies the (strong) coherence property as long as $M$ is not too small.

\subsection{Kerdock Codes}
The set of Kerdock codewords is given as columns of the matrix $\tilde{\mPsi}\in\{\pm 1, \pm j\}^{P\times P^2}$, where $P=2^m$. Since the Kerdock code is a cyclic extended code over $\mathbb{Z}^{4}$, we can find a map of the columns of Kerdock code into $P$ blocks (Theorem 10, \cite{CalderbankZ4}), such that the $P-1$ columns within each block are cyclic. Then we can assign adjacent codewords in the same cyclic block to one user, and set the first code as user's transmitted codeword. We denote the final code book as $\mPsi\in\{\pm 1, \pm j\}^{P\times (P^2-P)}$, and we denote the discarded $P$ columns by the set $\mPsi_c$.

The coherence property of the subsampled Kerdock code set is summarized in the Proposition below.
\begin{prop} \label{prop_kerdock}
Let $\X$ be a unit-column matrix with $M$ rows subsampled uniformly at random from a Kerdock code $\mPsi$.  If $M\geq \gamma \log^3 P$ for some constant $\gamma$, then with probability at least $1-2P^{-1}$, we have $\mu(\X)\leq \gamma_2/\log P$ for some constant $\gamma_2$, and $\nu(\X)\leq 2/M$ deterministically.

\end{prop}

\begin{proof} See Appendix~\ref{proof_kerdock}.
\end{proof}

The worst-case coherence of $\mPsi$ meets the Welch bound $\mu(\mPsi)=1/\sqrt{P}$ and the average coherence of $\mPsi$ is $\nu(\mPsi)=1/P$. Proposition~\ref{subsample_coherence} implies that we can find an $M$ such that the subsampled Kerdock code set satisfies the (strong) coherence property as long as $M$ is not too small.

%% file: numerical.tex
\section{Numerical Examples}\label{sec:numerical} 

\subsection{Gabor Signature Waveforms}
We first consider when each circulant matrix in the Gabor frame supports only one user. This corresponds to the maximum delay the algorithm can work in the asynchronous case. Let the seed vector $\g$ for the Gabor frame be either an Alltop sequence of length $P=127$, given as
$$ \g = \frac{1}{\sqrt{P}}[e^{j2\pi \frac{1^3}{P}}, e^{j2\pi \frac{2^3}{P}}, \ldots, e^{j2\pi \frac{P^3}{P}}];$$
or a unit vector with random uniform phase of length $P=128$, given as
\begin{equation} \label{randomseed}
 \g = \frac{1}{\sqrt{P}}[e^{j2\pi \theta_1},e^{j2\pi \theta_1},\ldots,e^{j2\pi \theta_P} ],
 \end{equation}
where $\theta_i$ is uniformly distributed on $[0,1]$, $1\leq i\leq P$. The power profile is assumed known as $r_n=1$ for all $n=1,\cdots, N$ in the coherent case, and are assume unknown in the noncoherent case.

The active users are selected first by uniformly choosing a number at random from $1$ to $P$, and then, for each active user, the delay is chosen uniformly at random. First, we fix the number of active users, namely $K=2$, and apply the coherent detector described in Algorithm 1 and noncoherent detector described in Algorithm 2 for $\SNR=20$dB and $\SNR=40$dB. The partial DFT matrix is applied with randomly selected rows and the number of Monto Carlo runs is $5,000$. Fig.~\ref{measurement} shows the probability of error for multi-user detector with respect to the number of measurements. The performance of the Alltop Gabor frame is better  than that of the random Gabor frame due to its optimal coherence. It is also worth noting that the performance of the noncoherent detector is almost the same as that of the coherent detector, albeit it does not perform symbol detection. This may suggest that channel state information and power control are less important in sparse recovery of active users.  

Finally, we consider when the maximum delay is relatively small, for example $\tau=15$ when $P=128$ for a random Gabor frame. We transmit the first sequence within the block of the circulant matrix, resulting in a total number of $P^2/(\tau+1)=1024$ users, and Fig.~\ref{more_user_coh} and Fig.~\ref{more_user_omp} show the probability of error for multi-user detection with respect to the number of active users $K$ for different number of random measurements $M=40, 60, 80$ when $\SNR=20$dB and $\SNR=40$dB respectively. 

\begin{figure*}[htp]
\centering
\begin{tabular}{cc}
\includegraphics[width=0.45\textwidth]{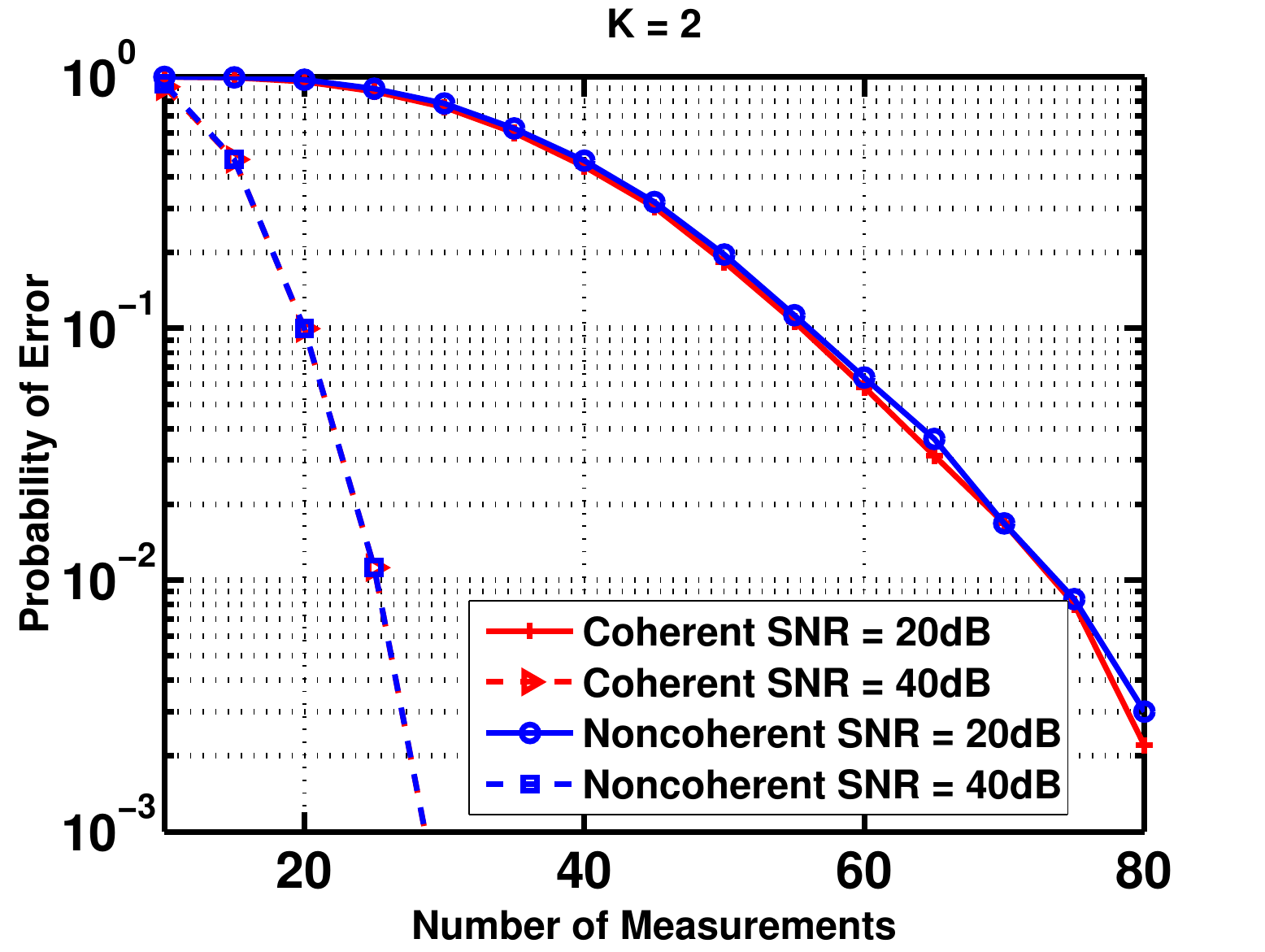}  &
\includegraphics[width=0.45\textwidth]{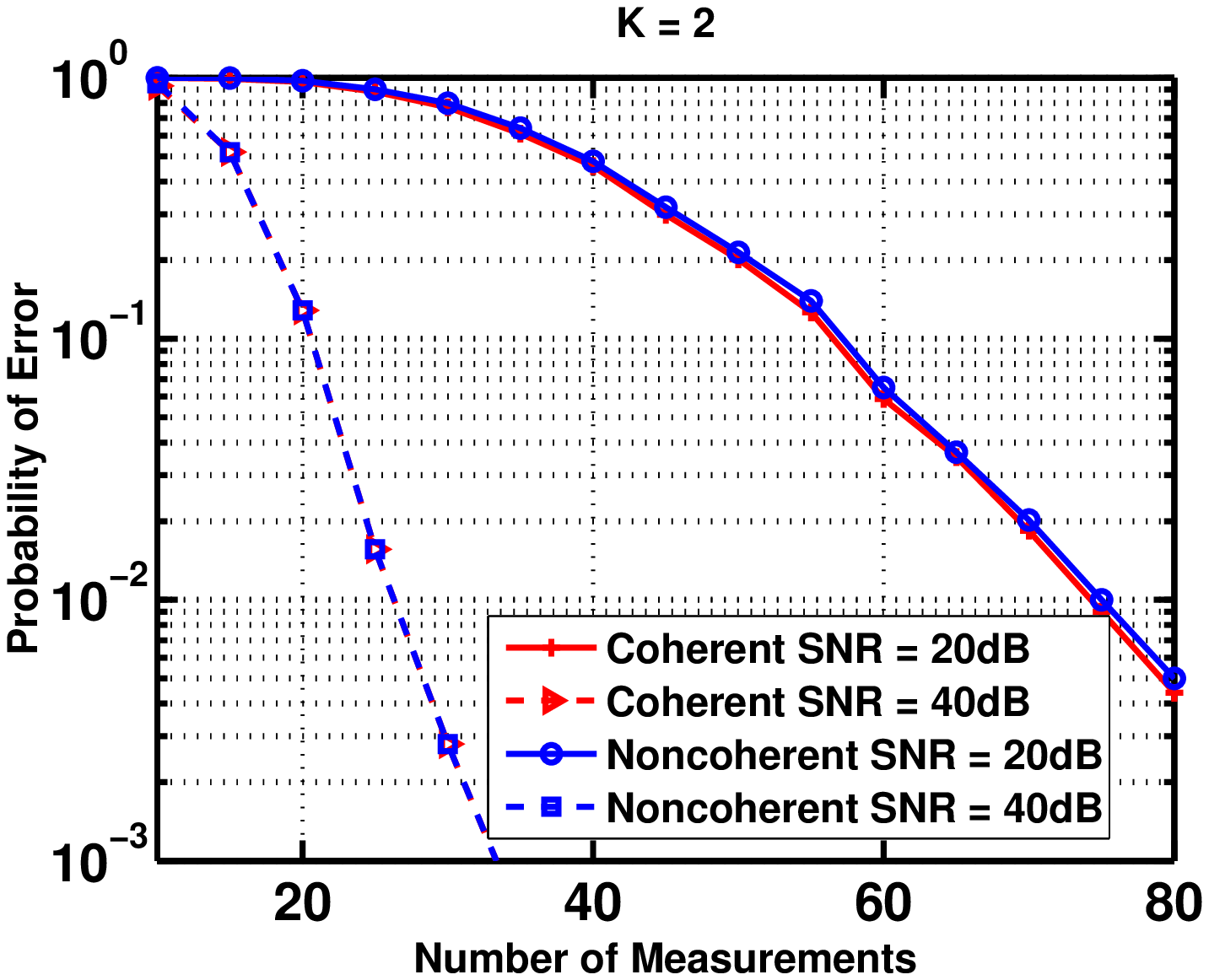} \\
(a) Alltop Gabor Frame &  (b) Random Gabor Frame 
\end{tabular}
\caption{Probability of error for multi-user detection with respect to the number of measurements from coherent and noncoherent detectors using (a) an Alltop Gabor frame with length $P=127$, and (b) a random Gabor frame with length $P=128$, for $K=2$ active users and $\SNR=20$dB and $40$dB, where the maximum chip delay is $\tau=126$. }\label{measurement}
\end{figure*}

\begin{figure*}[htp]
\centering
\begin{tabular}{cc}
\includegraphics[width=0.45\textwidth]{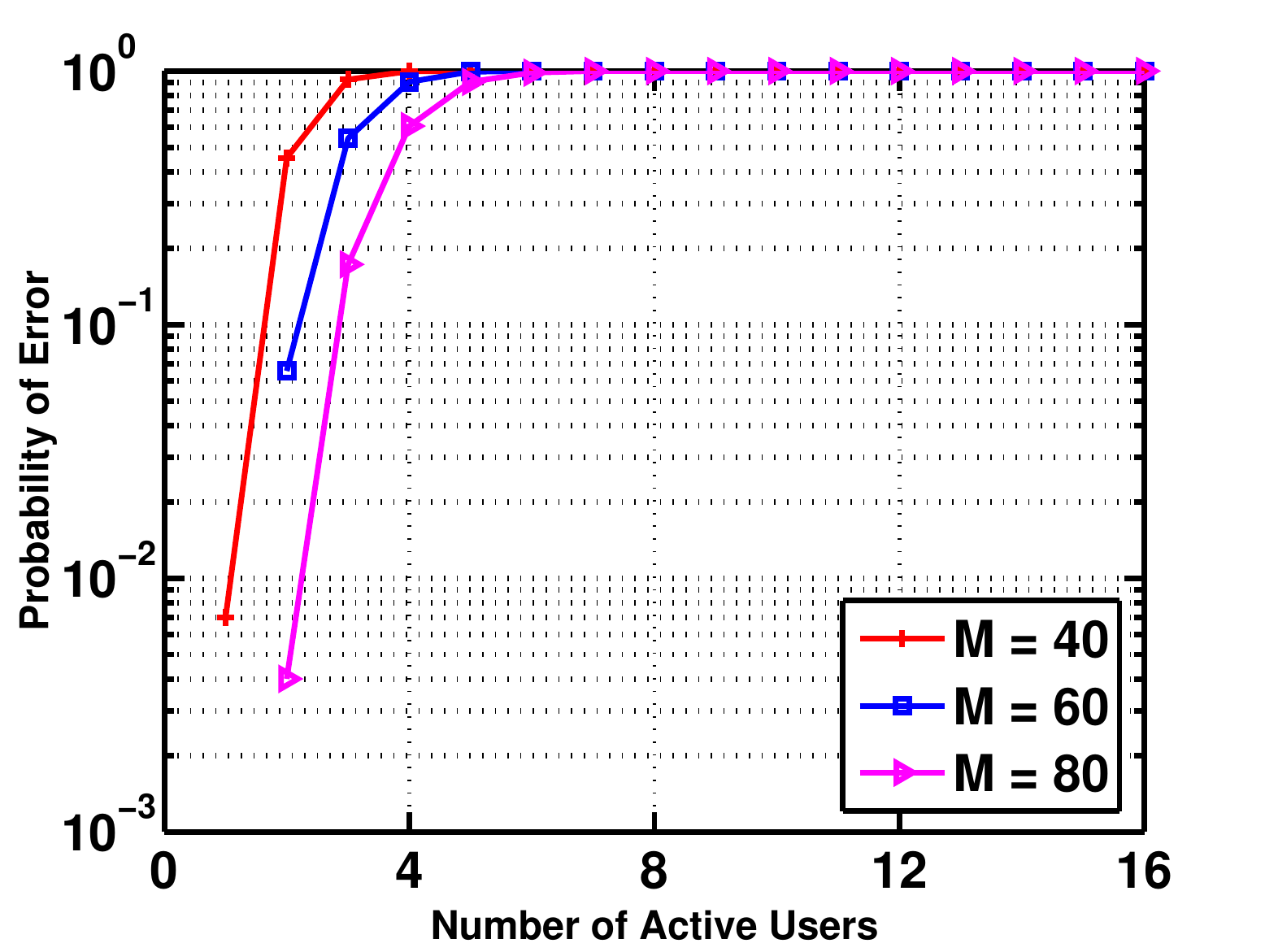}  &
\includegraphics[width=0.45\textwidth]{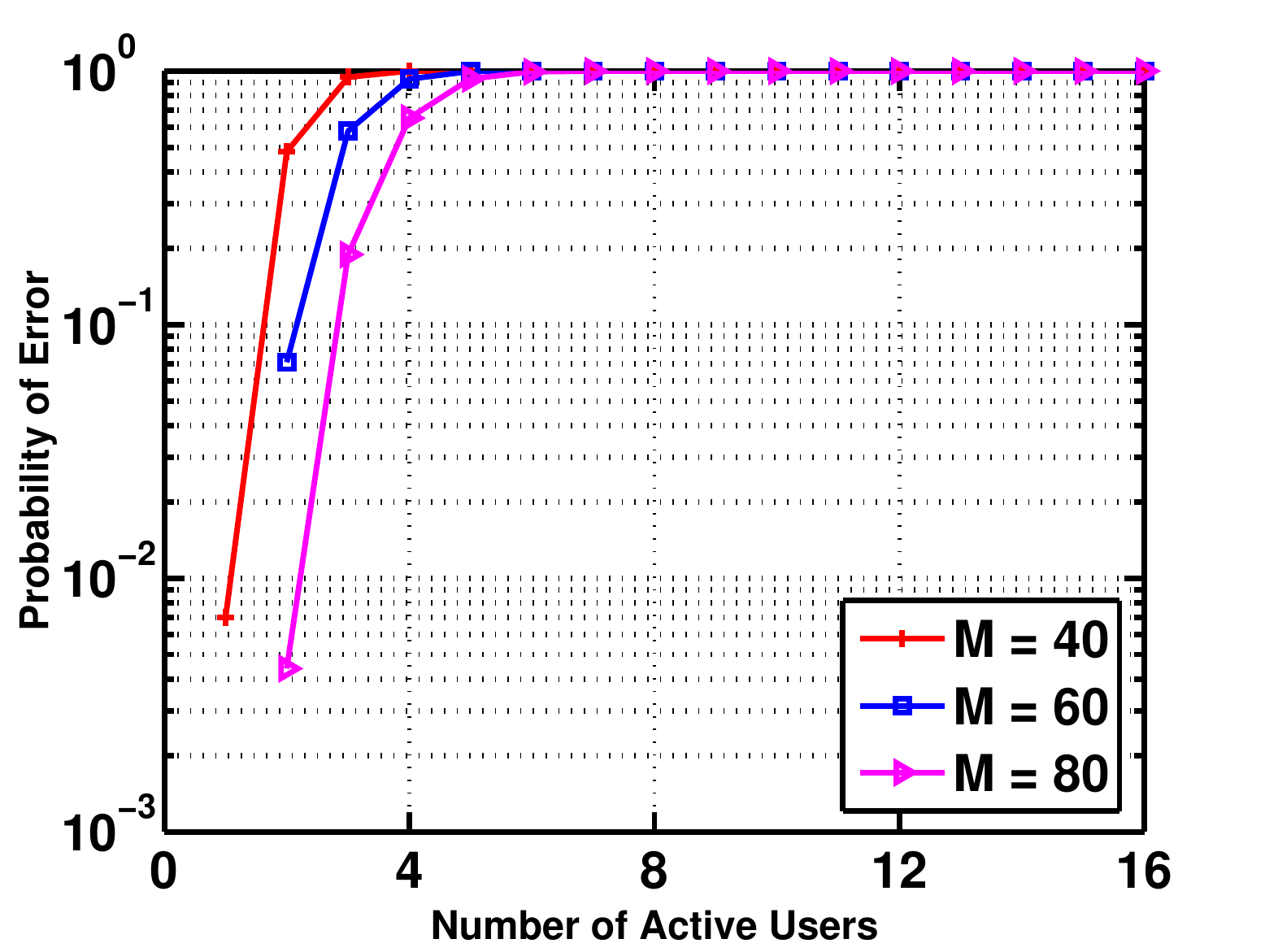}  \\
 (a) Coherent detector & (b) Noncoherent detector
 \end{tabular}
\caption{Probability of error for coherent and noncoherent multi-user detection with respect to the number of active users using a random Gabor frame with $P=128$ for $M=40, 60, 80$ when $\SNR=20$dB, where the maximum chip delay is $\tau=15$. The total number of users is $N=1024$.} \label{more_user_coh}
\end{figure*}

\begin{figure*}[htp]
\centering
\begin{tabular}{cc}
\includegraphics[width=0.45\textwidth]{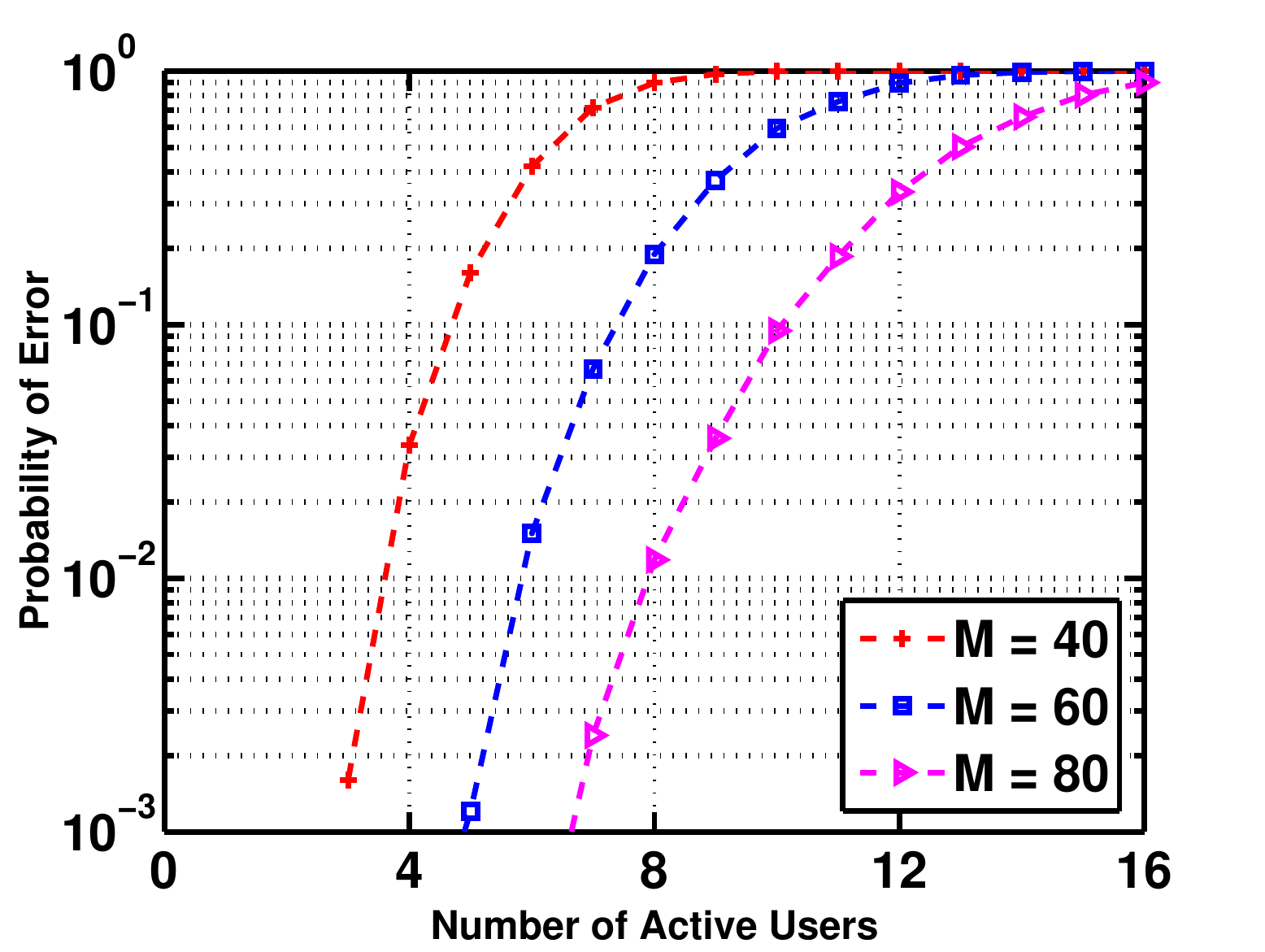}  &
\includegraphics[width=0.45\textwidth]{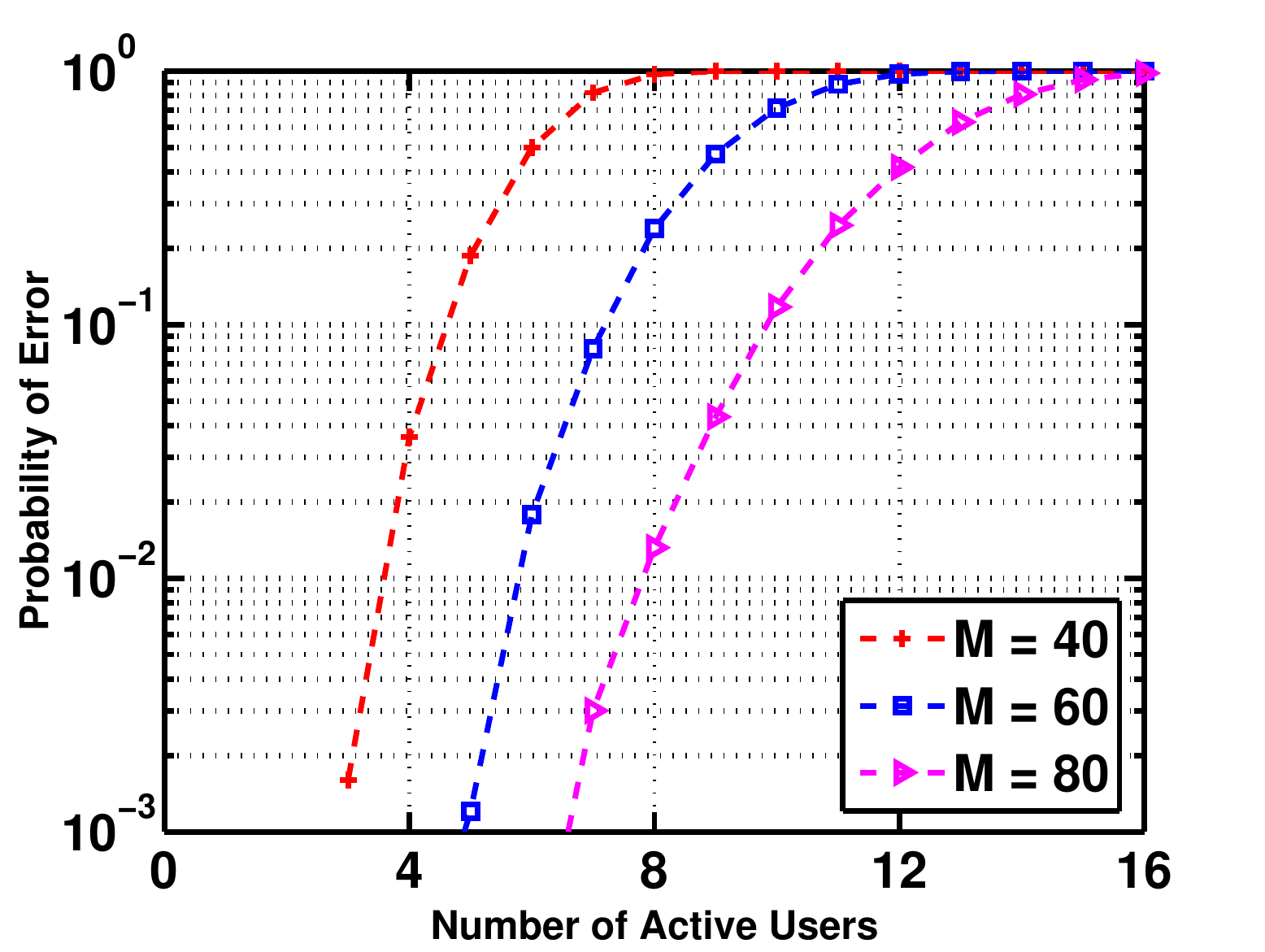}  \\
 (a) Coherent detector & (b) Noncoherent detector
 \end{tabular}
\caption{Probability of error for coherent and noncoherent multi-user detection with respect to the number of active users using a random Gabor frame with $P=128$ for $M=40, 60, 80$ when $\SNR=40$dB, where the maximum chip delay is $\tau=15$. The total number of users is $N=1024$.} \label{more_user_omp}
\end{figure*}

\subsection{Kerdock Signature Waveforms}
We first generate a Kerdock code set $\mPsi$ of length $P=128$ with $P^2$ codewords. By removing the all-one row in $\mPsi$, and removing two column in each block of size $P$, we obtain a block-circulant matrix of size $(P-1)\times P(P-2)$, where there are $P$ circulant blocks of size $(P-1)\times (P-2)$. As earlier, we assume the maximal delay is $\tau=15$, the total number of users is given as $\lfloor (P-2)/(\tau+1)\rfloor \cdot P=896$. Fig.~\ref{kerdock_moreuser} show the probability of error for multi-user detection with respect to the number of active users $K$ for different number of random measurements $M=20, 40, 60$ when $\SNR=20$dB. 

\begin{figure*}[htp]
\centering
\begin{tabular}{cc}
\includegraphics[width=0.45\textwidth]{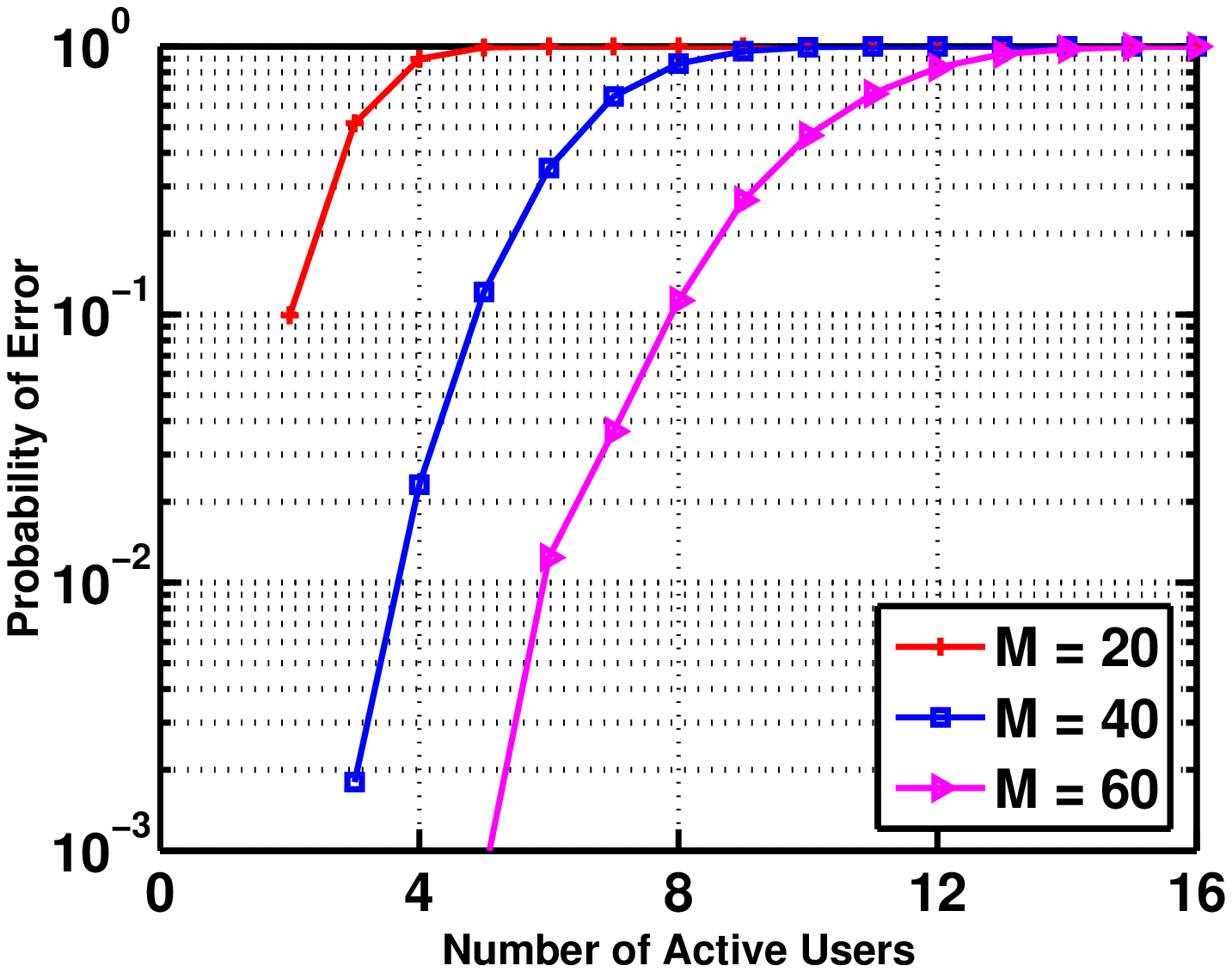}  &
\includegraphics[width=0.45\textwidth]{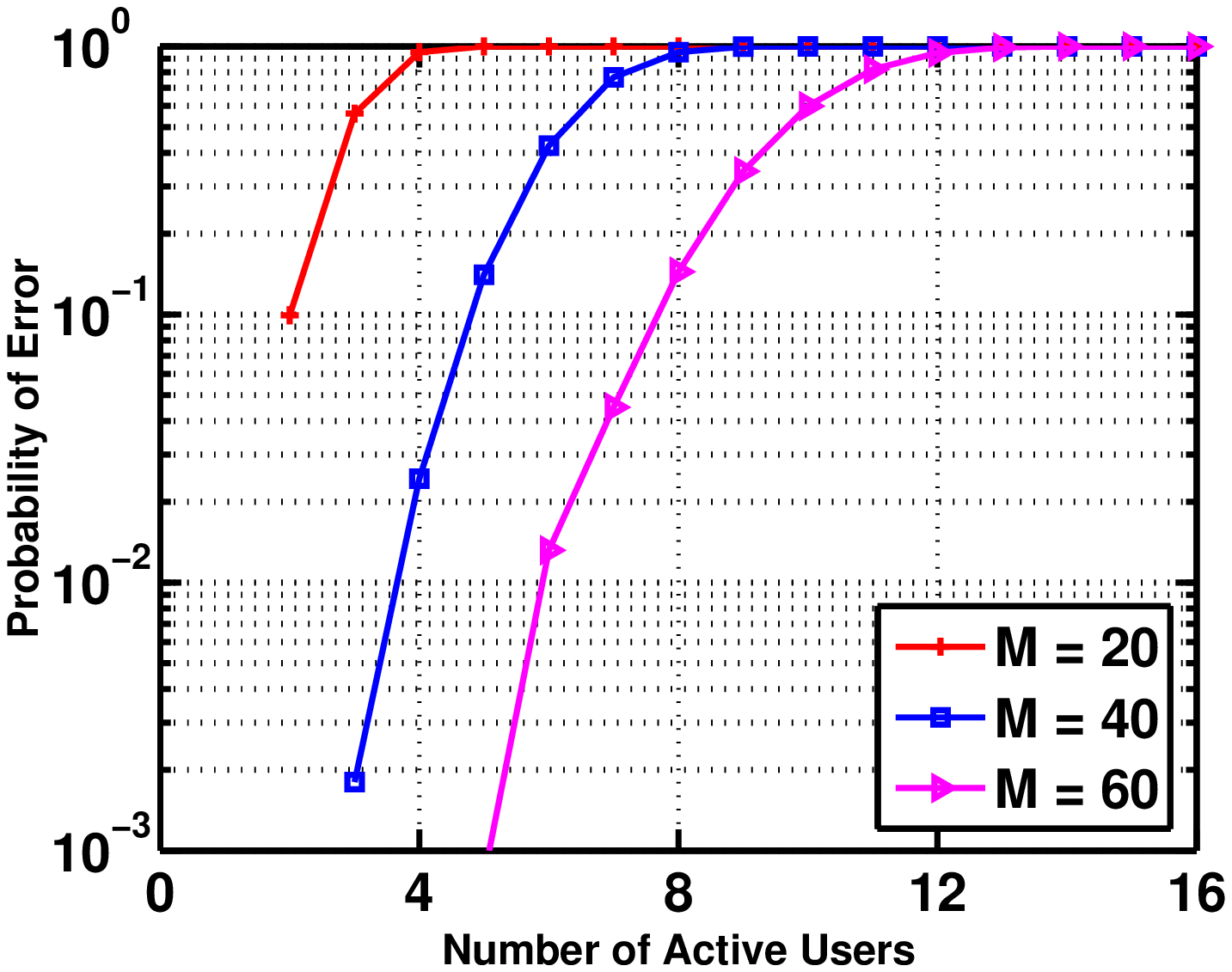}  \\
 (a) Coherent detector & (b) Noncoherent detector
 \end{tabular}
\caption{Probability of error for coherent and noncoherent multi-user detection with respect to the number of active users $K$ using a Kerdock code set with $P=127$ for $M=20,40,60$ when $\SNR=20$dB, where the maximum chip delay is $\tau=15$. The total number of user is $N=896$.} \label{kerdock_moreuser}
\end{figure*}

\subsection{Comparison of Signature Waveforms}

In this section we compare the performance of different signatures for multi-user detector when $\SNR=20$dB and $K=2$. We use the above considered Kerdock code, Alltop Gabor frame and random Gabor frame when $P=128$. We also consider the cyclic extensions of random matrix whose columns are generated from \eqref{randomseed}. Table~\ref{totaluser} summarizes the total number of users for different signature waveforms, notice that both Kerdock and Alltop suffer from the floor operation in calculating the number of total users. As shown in Fig.~\ref{compare}, the performance of Kerdock code is significantly better than other choices. The performance of cyclic extensions of random matrices and Gabor frames are similar, since the subsampling degenerates the optimal coherence properties of the unsampled Gabor frame. The Alltop Gabor frame is slightly better than its random counterparts.

\begin{table}
\centering
\begin{tabular}{|c|c|}
\hline
Signatures & \# of total users  \\ \hline
Kerdock & 896\\ \hline
Random Block & 1024\\ \hline
Alltop Gabor & 889\\ \hline
Random Gabor &1024 \\ \hline
\end{tabular}
\caption{Total number of users for different signatures.} \label{totaluser}
\end{table}

\begin{figure*}[htp]
\centering
\begin{tabular}{cc}
\includegraphics[width=0.45\textwidth]{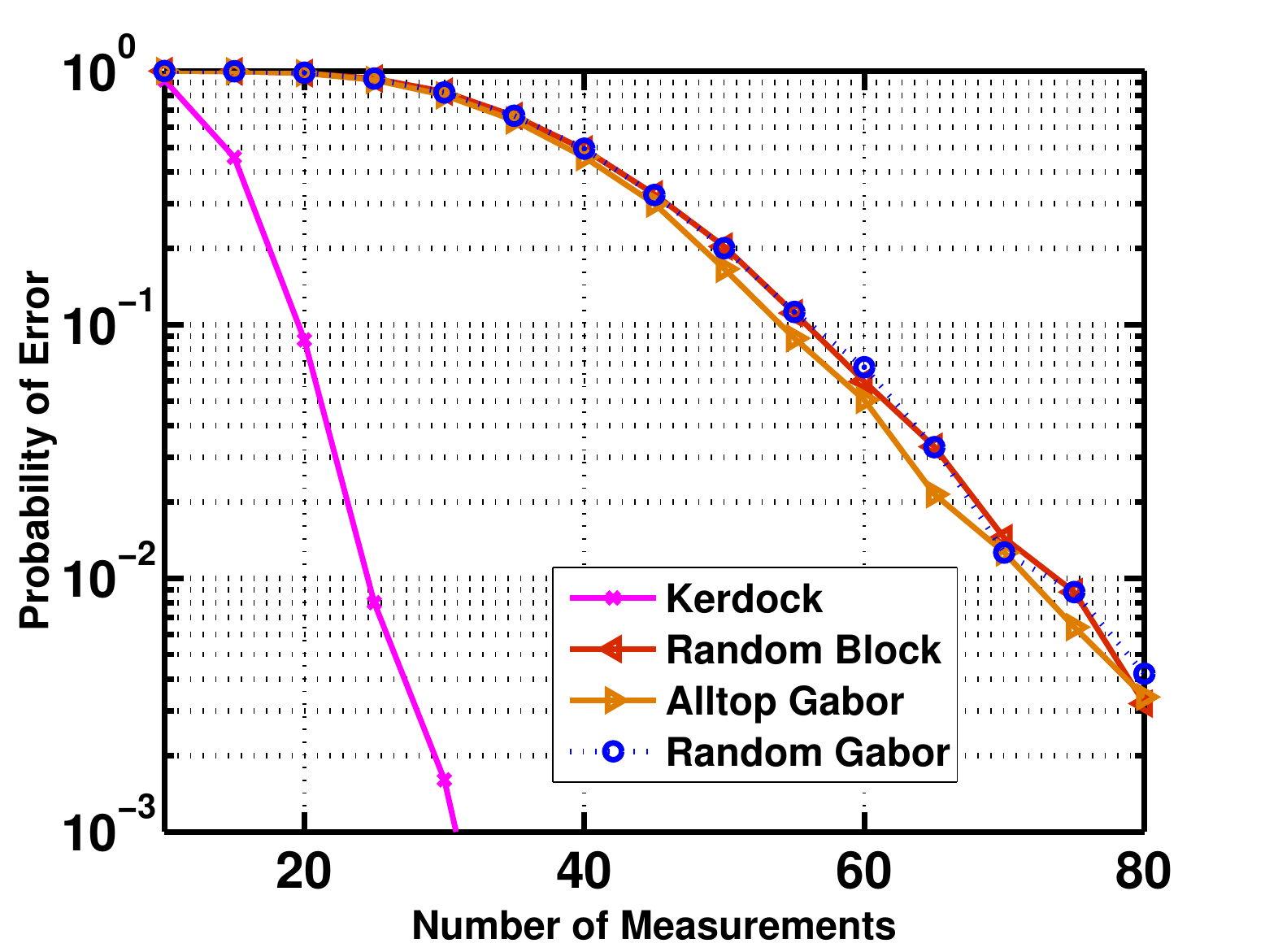}  &
\includegraphics[width=0.45\textwidth]{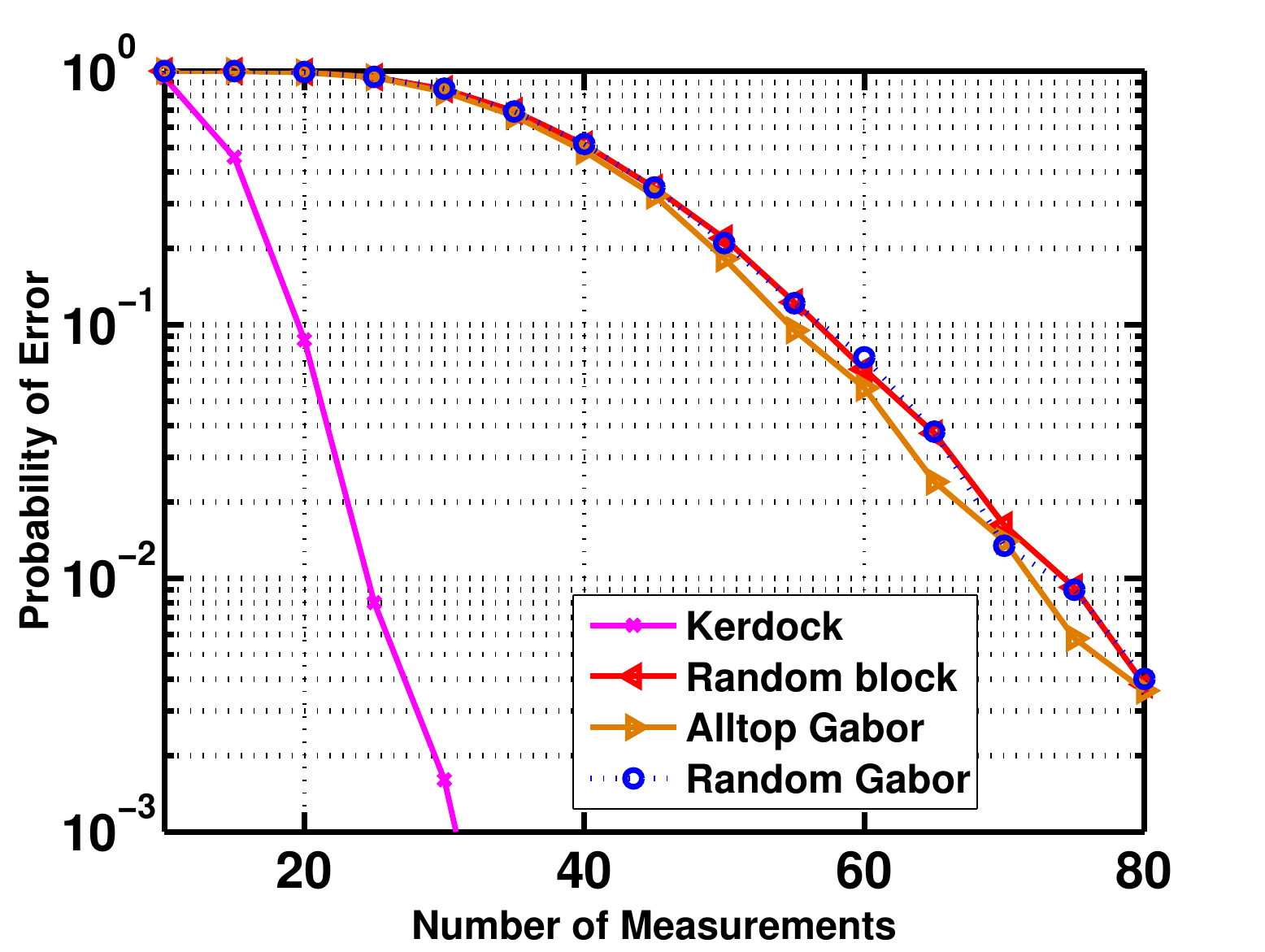}\\
 (a) Coherent detector &(b) Noncoherent detector 
 \end{tabular}
\caption{Comparison of performance with respect to the number of measurements for multi-user detection when $K=2$ and $\SNR=20$dB, where the maximum chip delay is $\tau=15$.} \label{compare}
\end{figure*}